\documentclass[aps,prd,nofootinbib,english,10pt]{revtex4-2}
\usepackage[utf8]{inputenc}
\usepackage[T1]{fontenc}
\usepackage{graphicx,color} 
\usepackage{xcolor}
\usepackage{rotating}
\usepackage{amssymb}
\usepackage{amsfonts}
\usepackage{amsmath}
\usepackage{amsthm}
\usepackage{tensor}
\usepackage{pifont}       
\usepackage{array}
\usepackage{xspace}
\usepackage{mathtools} 
\usepackage{verbatim}
\usepackage{alphabeta}
\usepackage{comment}
\usepackage{enumitem}
\usepackage{physics}
\usepackage[colorlinks]{hyperref}
\usepackage{MnSymbol}
\usepackage{accents}

\newtheorem{theorem}{Theorem}[section]

\newtheorem{definition}[theorem]{Definition}

\newtheorem{lemma}[theorem]{Lemma}

\begin{document}

\title{On the Finsler variational nature of autoparallels in metric-affine geometry
}


\author{Lehel Csillag}
\email{lehel.csillag@unitbv.ro}
\author{Nicoleta Voicu}
\email{nico.voicu@unitbv.ro}
\affiliation{Faculty of Mathematics and Computer Science, Transilvania University, Iuliu Maniu Street 50, Brașov 500091, Romania}

\author{Salah Elgendi}
\email{salah.ali@fsc.bu.edu.eg, \ salahelgendi@yahoo.com}
\affiliation{Department of Mathematics, Faculty of Science\\
Islamic University of Madinah, KSA\\
Department of Mathematics, Faculty of Science\\
Benha University, Egypt}

\author{Christian Pfeifer}
\email{christian.pfeifer@zarm.uni-bremen.de}
\affiliation{ZARM, University of Bremen, 28359 Bremen, Germany.}

\begin{abstract}
In metric–affine geometry, autoparallels are generically non-variational, i.e.\,, they are not the extremals of any action integral. The existence of a parametrization-invariant action principle for autoparallels is a long-standing open problem, which is equivalent to the so-called Finsler metrizability of the connection -- that is, to the fact that these autoparallels can be interpreted as Finsler geodesics. 

In this article, we address this problem for the class of torsion-free affine connections with vectorial nonmetricity, which includes, as notable subcases, Weyl and Schrödinger connections. For this class, we determine the necessary and sufficient conditions for the existence of a Finsler Lagrangian that metrizes the connection (and depends only algebraically on the metric and on the nonmetricity defining vector field). In the cases where such a Finsler Lagrangian exists, we construct it explicitly. In particular, we show that a broad class of such connections is in fact Finsler metrizable, i.e.\,, the autoparallels of these connections are Finsler geodesics.

\end{abstract}

\maketitle

\tableofcontents

\newpage

\section{Introduction}
In general relativity, which is based on pseudo-Riemannian geometry, freely falling particles follow curves that are simultaneously extremals of the length functional defined through the spacetime metric and  autoparallels of the Levi-Civita connection. The variational property of the length functional provides a natural action principle for these autoparallels, making the connection between geometry and dynamics explicit. 

Going beyond general relativity, when one considers spacetime manifolds with connections that are not metric compatible, so-called  metric-affine theories of gravity \cite{Iosifidis:2021pta, Blagojevic:2013xpa,Hehl:1994ue, Puetzfeld:2004yg}, this equivalence no longer holds. In a generic metric-affine geometry, the metric and  the affine connection are independent, giving rise to two inequivalent notions of preferred curves: geodesics, defined as extremals of the metric length functional, and autoparallels, defined as curves whose tangent vectors are parallel transported by the connection. In general, these affine connection autoparallels do not extremize any length functional, which raises questions about their possible physical interpretation.
Actually, the question of which class of preferred curves should be interpreted as free-fall trajectories -- metric geodesics, autoparallels, or even an entirely different class -- has a long history and is still under debate in metric-affine gravity~\cite{HEHL1971225,PhysRevD.104.044031, Obukhov2015,IOSIFIDIS2024138498, Iosifidis:2025ano}. Much of this debate stems precisely from the typical lack of variationality of autoparallels.

Yet, we already know, see, e.g., \cite{Fuster:2020upk, Cheraghchi:2022zgv, Voicu_2024} that, at least for \textit{some} torsion-free connections with nonmetricity, autoparallels do still arise from a parametrization-invariant action. It turns out that this action is nothing but a Finslerian length measure for curves, which thus naturally links Finsler geometry and metric-affine geometry\footnote{Actually, even a stronger result holds, \cite{Szenthe}: if an ODE system on an analytic manifold admits an analytic Lagrangian, then it admits a 2-homogeneous one in the velocities -- which, if nondegenerate, is nothing but a Finsler Lagrangian. That is, in quite some cases, the request of parametrization invariance can even be removed.}.
Moreover, we have shown that only for very special subclasses of connections with nonmetricity with respect to an initially considered pseudo-Riemannian metric, there exists a second, 'effective' pseudo-Riemannian metric whose Levi-Civita connection coincides with the originally considered connection. Generically, a Finsler structure is needed.

As Finsler geometry is, by definition, the geometry of a manifold equipped with the most general notion of parametrization-invariant arc-length functional, we argue that this is the appropriate setting for discussing the existence of a parametrization-invariant action for autoparallels. 


This feature adds nicely to the recent developments in the application of Finsler geometry to gravitational physics, as a very promising candidate for a model of the gravitational field of kinetic gases \cite{Hohmann:2019sni,Hohmann:2018rpp} capable of providing a geometric explanation for the dark energy phenomenology in cosmology \cite{Heefer:2023a,Friedl-Szasz:2024vtu,Pfeifer:2025tda}, or for its applications to quantum gravity phenomenology \cite{Addazi:2021xuf,Amelino-Camelia:2014rga,Lobo:2020qoa}.

In this work, we address the question of the existence of a parametrization-invariant action for autoparallels of symmetric affine connections \textit{with vectorial nonmetricity}. These are defined by the property that their nonmetricity tensor is given by an algebraic expression of the metric components and of the components of a given one-form.  This class of connections includes as notable subcases the Weyl, Schrödinger and completely symmetric geometries. Weyl's connection stands out, as it is the only conformally invariant affine connection, and it has been recently used to construct gauge theories of gravity with Standard Model matter \cite{Ghilencea2023,GHILENCEA2024169562, CONDEESCU2025170125}. Schrödinger connections have the appealing property of preserving lengths of autoparallels \cite{csillag,Csillag:2025gnz, PhysRevD.109.024003, Schrodinger}, even in the absence of metric compatibility, making them evade Einstein's objections to Weyl geometry. Completely symmetric connections are interesting, since a theory linear in the Ricci scalar within this geometry has been recently proven \cite{Csillag:2025gnz} to be on-shell equivalent with Scherrer's influential kinetic k-essence \cite{Scherrer:2004au}. 

For the class of torsion-free affine connections with vectorial nonmetricity, we find an answer to the following questions:
\begin{enumerate}
\item What are the necessary and sufficient conditions to be satisfied by a given torsion-free connection with vectorial nonmetricity, such that it is \textit{Finsler metrizable}, i.e., its autoparallels arise as arc-length parametrized  geodesics of a Finsler Lagrangian?
\item What are the corresponding Finsler Lagrangians, which depend algebraically on the components of the metric and the one-form defining the affine connection?
\end{enumerate}

Finsler Lagrangians whose geodesics arise as autoparallels of an affine connection on spacetime are known under the name of \textit{Berwald-type}~\cite{szilasi2011waysberwaldmanifolds,Pfeifer:2019tyy,shen_unicorns,Berwald1926,Cheraghchi:2022zgv} ones. 
On the other hand, Finsler Lagrangians that depend algebraically only\footnote{This ansatz is, at least intuitively, justified by the remark that the canonical (typically, nonlinear) connection of a Finsler space depends differentiably on the Finsler Lagrangian, and not vice-versa.} on the components of a pseudo-Riemannian metric $a$ and a one-form $b$ are called of \textit{generalized $(\alpha,\beta)$-type,} (where $\alpha = \sqrt{|a_{\mu\nu}\dot x^\mu \dot x^\nu|}$ stands for the metric pseudo-norm of vectors $\dot x$ on spacetime and $\beta = b_\mu \dot x^\mu$ is the contraction of the one-form with a vector $\dot x$); these depend on a specific combination of $\alpha$ and $\beta$  and in addition they can depend on the metric norm of the one-form $|b|=\sqrt{|a^{\mu\nu} b_\mu b_\nu|}$. \\
In other words, our two above questions can be reformulated, in technical terms, as: \textit{Classify torsion-free affine connections with vectorial nonmetricity, that are metrizable by Berwald-type $(\alpha,\beta)$-Finsler structures. }

Thus, to answer the above questions, we proceed in two steps:
\begin{enumerate}
\item We first consider the so-called \textit{$(\alpha,\beta)$-Finsler Lagrangians} - which are the simplest and most commonly used Finsler functions in applications. They are singled out in the class of generalized $(\alpha,\beta)$-Finsler Lagrangians as they do not depend on $|b|$. For this class, we completely clarify under which conditions $(\alpha,\beta)$-Finsler functions are of Berwald type and their relation to affine connections with vectorial nonmetricity.
In previous works this classification has been obtained in positive-definite signature \cite{shen_unicorns}, and some partial classification is known in Lorentzian signature, namely for the specific subclass of \textit{generalized m-Kropina metrics} \cite{Heefer:2022sgt, Fuster:2020upk, 10.1063/5.0227285, Heefer:2024kfi}. 

A complete discussion of how these results relate  to vectorial nonmetricity, together with the specific conditions on the given affine connection -- and hence their relevance for metric-affine gravity and cosmology -- has not been explored until now and will be first presented here. 

\item After having presented a clear discussion of the simpler case of \textit{$(\alpha,\beta)$}-Finsler Lagrangians, we pass to the most general case.  More precisely, for a given connection with vectorial nonmetricity, we find the necessary and sufficient conditions such that it is metrizable by a \textit{generalized  $(\alpha,\beta)$}-Finsler Lagrangian.
With this, we identify for the first time that a quite large subclass of affine connections with vectorial nonmetricity  can be understood as canonical connections of a Finsler Lagrangian of Berwald type, given the defining one-form satisfies appropriate constraints. In particular, all these connections have variational autoparallels.
\end{enumerate}

The paper is structured as follows. In Section \ref{sec:VecNonMetGeom} we briefly review symmetric affine connections and specialize to the case of vectorial nonmetricity. Section \ref{ssec:berw} introduces the necessary tools from Finsler geometry and formally defines the notions of Finsler and pseudo-Riemannian metrizability. In Sections \ref{sec:(alpha,beta)} and \ref{sec:generalized_ab}, we discuss the metrizability of connections with vectorial nonmetricity by $(\alpha,\beta)$-metrics and respectively, by generalized $(\alpha,\beta)$-metrics. By reformulating the metrizability problem as a first-order PDE system, we derive necessary and sufficient conditions, leading to complete classifications, in both cases.
We conclude in Section \ref{sec:Discussion} by summarizing our results and outlining possible directions for future research.

\section{Preliminaries}
This section is a minimalistic review of the known notions and results to be used in the following discussion of this article, related, on one hand, to affine connections with nonmetricity in metric-affine geometry and, on the other hand, to Finsler geometry. 

In the following,  by $M$ we will always mean a connected 4-dimensional smooth manifold, equipped with local coordinates $x^{\mu}, \, \mu= 0,\dots,3 .$ 

\subsection{Symmetric affine connections with vectorial nonmetricity}\label{sec:VecNonMetGeom}
Let $a$ be a Lorentzian metric and $\nabla$ a torsion-free\footnote{In this paper, “torsion-free affine connection” and “symmetric affine connection” are considered equivalent and will be used interchangeably.} affine connection on $M$. In any local chart the connection coefficients can be decomposed as \cite{Csillag:2025gnz}
\begin{equation}\label{eq:connectiondecomposition}
    \Gamma^{\mu}{}_{\nu \rho}=\overset{\circ}{\Gamma}{}^{\mu}{}_{\nu \rho}+\tensor{D}{^\mu _\nu _\rho} \, .
\end{equation}

Here, $\overset{\circ}{\Gamma}{}^{\mu}{}_{\nu \rho}$ denotes the Christoffel symbols of the Levi-Civita connection constructed from the metric components $a_{\mu \nu}$
\begin{align}
    \overset{\circ}{\Gamma}{}^\mu{}_{\nu\rho}
    = \tfrac{1}{2} a^{\mu\lambda}
    \left(
        \partial_\nu a_{\lambda\rho}
        + \partial_\rho a_{\lambda\nu}
        - \partial_\lambda a_{\nu\rho}
    \right)\, .
\end{align}
The remaining contribution $D^{\mu}{}_{\nu \rho}$ is the distortion tensor, which is determined by the nonmetricity tensor
\begin{align}
    \tensor{Q}{_\mu _\nu _\rho}= - \nabla_{\mu} a_{\nu \rho} \, .
\end{align}
The nonmetricity tensor measures the failure of the metric to be covariantly constant, implying that lengths and angles are generally not preserved under parallel transport. In terms of $Q_{\mu \nu \rho}$, the distortion tensor takes the form
\begin{align}\label{eq:distortiondef}
    \tensor{D}{^\mu _\nu _\rho}=\frac{1}{2} \left( \tensor{Q}{_\nu _\rho ^\mu}+\tensor{Q}{_\rho ^\mu _\nu} -\tensor{Q}{^\mu _\nu _\rho} \right) \, .
\end{align}

Given an affine connection $\nabla$, autoparallel curves $\gamma:I \subseteq \mathbb{R} \to M$ are curves whose tangent vector is parallel transported along itself
\begin{equation}\label{eq:autoparallel}
    \nabla_{\dot \gamma} \dot \gamma=0 \, .
\end{equation}

\textbf{Variationality of autoparallels}. For a general symmetric affine connection, the autoparallel equations are typically non-variational -- that is, the system above cannot, in general, be obtained as the Euler-Lagrange equations 

\begin{equation}
    \frac{d}{d \tau} \frac{\partial L}{\partial \dot \gamma^{\mu}} - \frac{\partial L}{\partial \gamma^{\mu}}=0
\end{equation}
for any Lagrangian $L=L(x,\dot x)$, and where $\gamma^{\mu}=x^{\mu} \circ \gamma$ is the coordinate representation of a curve and $\dot \gamma^\mu$ is the derivative of $\gamma^\mu$ with respect to the parameter $\tau$. This lack of variationality of autoparallels raises questions about their physical interpretation.

An important exception is the Levi-Civita connection $\overset{\circ}{\nabla}$ associated with the metric $a_{\mu \nu}$. In this case, the autoparallel equation
\begin{equation}
    \overset{\circ}{\nabla}_{\dot \gamma} \dot \gamma=0
\end{equation}
is variational, as it is equivalent to the geodesic equation

\begin{equation}\label{eq:LCgeod}
    \ddot \gamma^{\mu} + \overset{\circ}{\Gamma}{}^{\mu}{}_{\nu \rho} \dot \gamma^{\nu} \dot \gamma^{\rho}=0 \, ,
\end{equation}
arising as the Euler-Lagrange equation of the metric Lagrangian $\overset{\circ}{L}(x,\dot x)=\sqrt{|a_{\mu \nu}(x) \dot x^{\mu} \dot x^{\nu}|}$.

As announced in the Introduction, actually, the most general case when a nondegenerate, parametrization-invariant action for autoparallels exists is the one of \textit{Finsler metrizable} connections, \cite{Voicu_2024}, to be discussed in this paper. 

\bigskip

\textbf{Connections with vectorial nonmetricity}. In the following, we will focus on an interesting class of torsion-free connections, whose nonmetricity tensor is completely determined by the metric and a single one-form (or equivalently, by a vector field). Such connections appear naturally in various extensions of general relativity \cite{Csillag:2025gnz, Zhao2022, Iosifidis:2022evi, Iosifidis:2018jwu,Iosifidis:2020gth,Andrei:2024vvy, Khodadi2025, PhysRevD.107.064008}, and have applications ranging from information geometry \cite{iosifidis} to  constraining metric-affine theories using symmetry  principles~\cite{Paci_2023,Mikura2025}. 

\begin{definition}A symmetric affine connection $\nabla$ on $M$ is said to have \textbf{vectorial nonmetricity} \cite{Csillag:2025gnz} if there exists a non-zero one-form $b=b_{\mu}(x) dx^{\mu}$ on $M$ and three constants $c_1,c_2,c_3 \in \mathbb{R}$, not all zero, such that the nonmetricity tensor takes the form
\begin{align}\label{eq:vectorialnonmetricity}
     Q_{\mu \nu \rho}=c_1 b_\mu a_{\nu \rho} + c_2 \left( b_\rho a_{\mu \nu} + b_{\nu} a_{\rho \mu} \right) + c_3 b_\mu b_\nu b_\rho  \, .
\end{align}
\end{definition}
\noindent Using equation \eqref{eq:distortiondef}, the corresponding distortion tensor becomes
\begin{equation}
D_{~\nu \rho }^{\mu }=\dfrac{1}{2}\left( 2c_{2}-c_{1}\right) b^{\mu }a_{\nu
\rho }+\dfrac{1}{2}c_{1}b_{\nu }\delta^{\mu}{}_{\rho}+\dfrac{1}{2}%
c_{1}b_{\rho }\delta^{\mu}{}_{\nu}+\dfrac{1}{2}c_{3}b^{\mu }b_{\nu }b_{\rho
} \, .  \label{eq:D_vector_nonmetricity}
\end{equation}

Within the class of connections with vectorial nonmetricity, three subclasses of particular interest arise for specific choices of the coefficients $c_1,c_2,c_3$, summarized in Table \ref{table:constraints} below.

\begin{table}[h!]
\centering
\begin{tabular}{|c|c|}
\hline
Geometry & Constraint on coefficients \\ \hline
Weyl & $c_2 = c_3 = 0$ \\ \hline
Schrödinger & $c_1 + 2 c_2 = 0, \quad c_3 = 0$ \\ \hline
Completely symmetric & $c_1 = c_2$ \\ \hline
\end{tabular}
\caption{Constraints on the coefficients $(c_1,c_2,c_3)$ that distinguish the main subclasses of vectorial nonmetricity.}
\label{table:constraints}
\end{table}

The choice of the coefficients $c_1,c_2,c_3$ determines how lengths, volumes, and angles change under parallel transport. In Weyl geometry, angles are preserved, while lengths are generally not, whereas in Schrödinger geometry autoparallels have fixed length, despite the presence of nonmetricity \cite{csillag}. Interestingly, the existence of a covariantly preserved volume form is equivalent to the symmetry of the Ricci tensor associated with a connection possessing vectorial nonmetricity. For a more detailed discussion of the geometric properties of these connections, we refer the reader to Appendix A of \cite{Csillag:2025gnz}. 

For connections with vectorial nonmetricity, the autoparallel equation reads
\begin{equation}
    \ddot \gamma^{\mu}+ \left( \overset{\circ}{\Gamma}{}^\mu{}_{\nu\rho} + b^\mu a_{\nu \rho} \left(\frac{2 c_2 - c_1}{2} \right) + \frac{c_1}{2} \delta^{\mu}{}_{\nu} b_{\rho}+\frac{c_1}{2} \delta^{\mu}{}_{\rho} b_{\nu} +\frac{c_3}{2} b^\mu b_\nu b_\rho \right) \dot \gamma^{\nu} \dot \gamma^{\rho}=0 \, . 
\end{equation}
As we will show, for a quite large class of such connections, the above equation admits a variational description, which can be achieved by passing to Finsler geometry.

\subsection{Finsler geometry} \label{ssec:berw}
This subsection briefly reviews the basic notions of (pseudo-)Finsler geometry
and, in particular, of Berwald-type geometry, that is needed in the following. It mainly relies on \cite{Voicu_2024}.~For more details, see~\cite{Bucataru, math-foundations} and references therein.

Denote with $\pi:TM \to M$ the tangent bundle. Given a local chart $(U, x^{\mu})$ on $M$, the naturally induced coordinates on $\pi^{-1}(U) \subset TM$ are $\left( x^{\mu}, \dot x^{\mu} \right)$, where $\dot x= \dot x^{\mu} \partial_{\mu} \in T_{x}M$. We will often write simply $\left( x^{\mu}, \dot x^{\mu} \right) \equiv (x, \dot x)$ when the chart is fixed and hence no confusion arises. The natural local basis of $T_{(x, \dot x)} TM$ is
\begin{equation}
    \left\{ \partial_{\mu}:= \frac{\partial}{\partial x^{\mu}}, \; \; \dot \partial_{\mu}:= \frac{\partial}{\partial \dot x^{\mu}} \right\} \, ,
\end{equation}
with dual basis of $T^{*}_{(x,\dot x)}TM$
\begin{equation}
    \left\{ dx^\mu, d \dot{x}^{\mu} \right\} \, .
\end{equation}

\begin{definition}\label{def:fins}
    A \textbf{(pseudo)-Finsler structure} \cite{BejancuFarran2000}  on $M$  is a smooth function $L:\mathcal{A} \to \mathbb{R}$ defined on a conic subbundle\footnote{A conic subbundle of $TM$ is an open subset $\mathcal{A} \subseteq TM \setminus \{ 0 \}$ such that the fiber $\mathcal{A}_{x}:=\mathcal{A} \cap T_{x}M$ is non-empty for all $x \in M$, and  $\mathcal{A}$ is stable under positive rescalings, that is, $(x,\dot x) \in \mathcal{A} \implies (x, \lambda \dot x) \in \mathcal{A}$ for all $\lambda>0$.} $\mathcal{A} \subseteq TM \setminus \{ 0 \}$, satisfying:
    \begin{enumerate}
    \item Positive 2-homogeneity: $L(x, \lambda \dot x)=\lambda^2 L(x, \dot x), \forall \lambda >0 .$
    \item Nondegeneracy: the Hessian $g_{\mu \nu}(x,\dot x)=\frac{1}{2} \frac{ \partial^2 L}{\partial \dot x^\mu \partial \dot x^\nu}(x, \dot x)$ is non-singular at all $(x,\dot x) \in \mathcal{A}$.
    \end{enumerate}
    The pair $(M,L)$ is called a \textbf{(pseudo)-Finsler space}.
\end{definition}
The above notion includes as subclasses \emph{classical Finsler spaces} (where $\mathcal{A}=TM \setminus \{0\}$ and $(g_{\mu\nu})$ is everywhere positive definite), and \textit{Finsler spacetimes} (where $(g_{\mu\nu})$ has Lorentzian signature on an appropriate subset of $\mathcal{A}$). In what follows, we will consider the most general case and, for simplicity, sometimes omit the prefix “pseudo-”.

Any Finsler structure $L$ extends continuously to $\dot x=0$ by $L(x,0)=0$. The conic subbundle $\mathcal{A}$ is the set of \textit{admissible vectors,} and the functions
    \begin{equation}\label{eq:finsmet2}
        g_{\mu \nu}(x,\dot x)=\frac{1}{2} \frac{\partial ^2 L}{\partial \dot x^\mu \partial \dot x^\nu}(x,\dot x)
    \end{equation}
    are the local components of the \textit{Finslerian metric tensor}, which is a well-defined mapping
    \begin{equation}
        g: \mathcal{A} \to T^{0}_{2}M, \; \; (x, \dot x) \mapsto g_{(x, \dot x)}=g_{\mu \nu}(x,\dot x) dx^{\mu} \otimes dx^{\nu}\, .
    \end{equation}

   A \textit{parametrized admissible curve} on a Finsler space $(M,L)$ is a smooth map 
    \begin{equation}
        \gamma:[a,b] \to M, \; \; \tau \mapsto \gamma (\tau) \, ,
    \end{equation}
    such that its tangent lift
    \begin{equation}
        \gamma^{TM}:[a,b] \to TM, \; \; \tau \mapsto (\gamma(\tau),\dot \gamma(\tau))\equiv\left(\gamma^{\mu}(\tau), \dot\gamma^\mu(\tau) \right)
    \end{equation}
    lies in $\mathcal{A}$. The Finslerian \textit{arc-length} of the admissible curve $\gamma$ is, by definition:
    \begin{equation}
      \ell_{\gamma}=\int_{\gamma^{TM}} ds:= \int_{a}^{b} \sqrt{|L\left(\gamma(\tau),\dot{\gamma}(\tau) \right) |} d \tau \, .
    \end{equation}

    \textit{Geodesics} of a Finsler structure $L$ are the critical points of the functional $\gamma \mapsto \ell_{\gamma}$ and, in arc-length parametrization, they satisfy 
    \begin{equation}
        \ddot\gamma^\mu+2 G^{\mu} \left(\gamma, \dot\gamma \right) =0 \, ,
    \end{equation}
     where the functions $G^{\mu}$, called the \textit{spray coefficients} are defined in each coordinate chart by
    \begin{equation}
        G^{\mu}(x,\dot x):=\frac{1}{4} g^{\mu \nu} \left( \dot{x}^\sigma \partial_{\sigma} \dot \partial_{\nu} L(x, \dot x) - \partial_{\nu}L(x, \dot x)\right) \, .
    \end{equation}



The spray coefficients $G^{\mu}$ give rise to a canonical Finslerian extension of the Levi-Civita connection -- which is typically a \textit{nonlinear connection} on $\mathcal{A} \subseteq {TM \setminus \{0\}}$, as defined below.

\begin{definition}
    A \textbf{nonlinear connection} $N$ on $\mathcal{A} \subseteq TM \setminus \{0\}$ is a smooth assignment
    \begin{equation}
       \mathcal{A} \ni (x,\dot x) \mapsto H_{(x, \dot x)} \mathcal{A} \subseteq T_{(x, \dot x)} \mathcal{A} \, ,
    \end{equation}
    where $H_{(x,\dot x)} \mathcal{A}$ is a $4$-dimensional horizontal subspace complementary to the vertical subspace
    \begin{equation}
        V_{(x, \dot x)} \mathcal{A}:=\ker d \pi_{(x, \dot x)}= Span\left\{ \dot \partial_{\mu} \right\} \, .
    \end{equation}
\end{definition}
\noindent Hence, at each point $(x, \dot x) \in \mathcal{A}$, the tangent space splits as
    \begin{equation}
        T_{(x, \dot x)} \mathcal{A}=H_{(x,\dot x)} \mathcal{A} \oplus V_{(x, \dot x)} \mathcal{A} \, ,
    \end{equation}
with the nonlinear connection giving a \textit{local adapted basis} of $T \mathcal{A}$
\begin{equation}
    \{ \delta_{\mu}:=\partial_{\mu}- G^{\nu}{}_{\mu} \dot \partial_{\nu}\, ;\,  \dot \partial_{\mu} \}, \; \; H \mathcal{A}=Span \left\{\delta_{\mu} \right\}, \; \; V \mathcal{A}=Span \left\{ \dot \partial_{\mu}  \right\} \, .
\end{equation}
The real-valued functions $G^{\nu}{}_{\mu}=G^{\nu}{}_{\mu}(x, \dot x)$ defined on $\pi^{-1}(U)$,  called \textit{nonlinear connection coefficients}, locally specify the connection.

\bigskip

In a Finsler space $(M,L)$ the spray coefficients canonically induce a nonlinear connection via
\begin{equation}\label{eq:geods}
    G^{\nu}{}_{\mu}(x, \dot x)=\dot {\partial}_{\mu} G^{\nu}(x, \dot x)\,  ,
\end{equation}
which is called the \textit{canonical nonlinear connection} of $(M,L)$. A key property, used in the following, is that the Finsler structure $L$ is horizontally constant with respect to this connection. In terms of the adapted basis, this condition reads
\begin{equation}
    \delta_{\mu} L=0 \, , \; \;  \text{or equivalently,} \; \; \partial_{\mu} L - G^{\nu}{}_{\mu} \dot{\partial}_{\nu} L=0 \; .
\end{equation}

In particular, if $(M,L)$ is pseudo-Riemannian, that is, $L=a_{\mu \nu}(x) \dot x^{\mu} \dot x^{\nu}$ is quadratic in $\dot x$, then
\begin{equation}
    G^{\nu}{}_{\mu}(x, \dot x)=\overset{\circ}{\Gamma}{}{^\nu}{}_{\mu \rho}(x) \dot x^{\rho}
\end{equation}
is linear in $\dot x$. Though, typically, $G^{\nu}{}_{\mu}(x, \dot x)$ are not linear, but just $1$-homogeneous in $\dot x$.

\bigskip
A special and for us most important class of Finsler spaces  consists of \textit{Berwald spaces}, which have nontrivial Finsler structures ($L$ non-quadratic in $\dot x$), yet, their canonical connection is linear. 

\begin{definition}
    A Finsler space $(M,L)$ is \textbf{of Berwald-type} if in one (and then in any) local chart its canonical spray coefficients are quadratic in $\dot x$.
\end{definition}
This is equivalent to any of the following conditions:
\begin{enumerate}
    \item $N$ descends into a well-defined \textit{canonical affine connection}  $\nabla$ on $M$, with local coefficients $\Gamma^{\mu}{}_{\nu \rho}$ given by
\begin{equation}
2G^{\mu}(x,\dot x)=\Gamma^{\mu}{}_{\nu\rho}(x)\dot x^{\nu}\dot x^{\rho}
\iff
G^{\mu}{}_{\nu}(x,\dot x)=\Gamma^{\mu}{}_{\nu\rho}(x)\dot x^{\rho}
\iff
\dot\partial_{\rho}G^{\mu}{}_{\nu}(x,\dot x)=\Gamma^{\mu}{}_{\nu\rho}(x) \, .
\end{equation}
    \item Its arc-length parametrized geodesics coincide with the autoparallels of a symmetric affine connection $\nabla $ on $M$, with coefficients $\Gamma^{\mu}{}_{\nu \rho}(x)$:
    \begin{equation}
        \ddot \gamma^\mu + \Gamma^\mu{}_{\nu\rho}(\gamma)\dot \gamma^\nu \dot \gamma^\rho=0 \, .
    \end{equation}
\end{enumerate}

Berwald-type Finsler spaces therefore provide a natural setting for addressing our main question:

\begin{center}
\textit{Given a symmetric affine connection with vectorial nonmetricity $\nabla$, does there exist a Finsler space $(M,L)$ of Berwald type, such that the geodesics of $(M,L)$ coincide with the autoparallels of $\nabla$?}
\end{center}
Some definitions are in place here.

\begin{definition}\label{def:PseudoRiemMetrizability}
    A symmetric affine connection $\nabla$ on a manifold $M$ is said to be:
    \begin{enumerate}
        \item[$(i)$] \textbf{(Pseudo)-Riemann-metrizable} if there exists a pseudo-Riemannian metric $a$ such that $\nabla$ coincides with the Levi-Civita connection $\overset{\circ}{\nabla}$ of $a$, that is, locally we have
    \begin{equation}
         \Gamma^\mu{}_{\nu\rho} = \overset{\circ}{\Gamma}{}^\mu{}_{\nu\rho}[a] \,.
    \end{equation}
    In this case, the pseudo-Riemannian metric $a$ is said to metrize $\nabla \,$.
    \item[$(ii)$] \textbf{(Pseudo)-Finsler-metrizable} if there exists a Berwald-type pseudo-Finsler structure $(M,L)$, whose canonical affine connection is $\nabla$. In this case, the Berwald-type pseudo-Finsler space $(M,L)$ is said to metrize $\nabla\,$.
    \end{enumerate}
\end{definition}

\noindent \textbf{Remarks.}
\begin{enumerate}
\item  Saying that  $\nabla$ is pseudo-Finsler metrizable is equivalent to saying that there exists a nondegenerate, positively 2-homogeneous Lagrangian $L=L(x,\dot{x})$ (i.e., a reparametrization-invariant arc length functional) whose arc-length parametrized geodesics are the autoparallels of  $\nabla$. 
\item  Pseudo-Riemann metrizability obviously implies pseudo-Finsler metrizability. Yet, the converse statement is more nuanced: whereas metrizability by a classical ($TM\setminus \{0\}$-smooth, positive definite) Finsler metric does imply Riemann metrizability (a result known as Szabo's Theorem, \cite{Szabo}), there exist multiple examples of connections that are metrizable by pseudo-Finsler spacetime functions $L$, yet are \textit{not} pseudo-Riemann metrizable, see, e.g., \cite{Fuster:2020upk, Cheraghchi:2022zgv, Voicu_2024}. These are explicit counterexamples to the claim in a very recent paper \cite{Heisenberg} that variationality of autoparallels equates pseudo-Riemann metrizability.
\end{enumerate}

There exist multiple ways of characterizing the pseudo-Finsler metrizability of a given affine connection, see, e.g.,~\cite{BucataruDahl2009, Pfeifer:2019tyy}.  In the following, we will use the computationally simplest one, introduced by Z. Muzsnay in \cite{Muzsnay2006} and recently used by us in \cite{Cheraghchi:2022zgv} and \cite{Voicu_2024}.  This is based on the statement that the Finsler metrizability of a symmetric affine connection $\nabla$ on a smooth manifold $M$ is equivalent to the existence of  a positively 2-homogeneous, nondegenerate solution of the PDE system
\begin{equation}\label{def:Berwald_conditions}
    \delta_{\mu} L=0, \; \; \mu=0,\dots,3 \, .
\end{equation}

\textbf{Pseudo-Finsler metrizability of connections with vectorial nonmetricity.} Given a pseudo-Riemannian metric $a = a_{\mu\nu}(x)dx^\mu dx^\nu$ and a nonzero one-form $b=b_{\mu}(x) dx^{\mu}$, we denote by
   \begin{equation}\label{eq:AB}
       A:=a_{\mu \nu}(x) \dot{x}^{\mu}\dot{x}^{\nu}, \; \; \; \;  B:= b_{\mu} \dot x^{\mu} \, 
   \end{equation}
their values on an arbitrary vector $\dot{x}\in T_xM$, at any $x\in M$. In the following, we will refer loosely to both $A$  and $a$, respectively, $B$ and $b$ as the Riemannian metric, respectively, the one-form.  
In the following, indices will be raised and lowered by $a$; consistently with this convention\footnote{Explicitly, we have $A=a_{\mu \nu } \dot x^{\mu} \dot x^{\nu}$, so its fiber derivative is $\dot{\partial}_{\mu} A=\dot\partial_\mu\left( a_{\sigma \nu} \dot x^\sigma \dot x^\nu \right)=2 a_{\mu \nu} \dot x^{\nu}=2 \dot x_{\mu}$.}, we also denote
\begin{equation}
2\dot{x}_{\mu }:= \dot{\partial}_{\mu} A.
\end{equation}

As stated in the previous subsection, see \eqref{eq:D_vector_nonmetricity}, a symmetric affine connection with vectorial nonmetricity $\nabla$ is locally given by 
\begin{align}\label{eq:affinePi}
    \Gamma^\mu{}_{\nu\rho}  
    = \overset{\circ}{\Gamma}{}^\mu{}_{\nu\rho} + D^\mu{}_{\nu\rho} 
    = \overset{\circ}{\Gamma}{}^\mu{}_{\nu\rho} + b^\mu a_{\nu \rho} \left(\frac{2 c_2 - c_1}{2} \right) + \frac{c_1}{2} \delta^{\mu}{}_{\nu} b_{\rho}+\frac{c_1}{2} \delta^{\mu}{}_{\rho} b_{\nu} +\frac{c_3}{2} b^\mu b_\nu b_\rho\,.
\end{align}
Passing to the tangent  $TM$, $\nabla$ induces a nonlinear connection with coefficients $G^{\mu}{}_{\nu}=\Gamma^{\mu}{}_{\nu \rho} \dot x^{\rho}$; that is, the horizontal basis elements are
\begin{equation}\label{eq:CoeffsDnumu}
    \delta_{\mu}:=\partial_{\mu}-\dot{x}^{\rho} \Gamma^{\nu}{}_{\mu \rho} \dot \partial_{\nu} \; = \overset{\circ}{\delta}_{\mu} - D^{\nu}{}_{\mu} \dot \partial_{\nu} \, ,  
\end{equation}
where  we have denoted:
\begin{equation} \label{def:D^mu_nu}
 D^{\nu}{}_{\mu}:=D^{\nu}{}_{\mu \rho}(x) \dot{x}^{\rho},  \qquad \overset{\circ}{\delta}_{\mu}=\partial_{\mu}-\overset{\circ}{\Gamma}{}^{\nu}{}_{\mu \rho} \dot x^{\rho} \dot{\partial}_{\nu}.
\end{equation}
In the specific case of vectorial nonmetricity, 
we obtain that the contracted distortion coefficients $D^{\mu}_{~\nu}$:
\begin{equation} \label{eq:D^mu_nu}
D_{~\nu }^{\mu } =D_{~\nu \rho }^{\mu }\dot{x}^{\rho }=b^{\mu }\dot{x}
_{\nu }\left( c_{2}-\dfrac{c_{1}}{2}\right) +\dfrac{c_{1}}{2}\left( \delta
_{~\nu }^{\mu }B+\dot{x}^{\mu }b_{\nu }\right) +\dfrac{c_{3}}{2}b^{\mu
}b_{\nu }B \,. 
\end{equation}
It further leads to the expressions, which we will need in the upcoming derivations
\begin{eqnarray}
D_{~\mu }^{\nu }\dot{x}_{\nu } &=&c_{2}B\dot{x}_{\mu }+\dfrac{1}{2}\left(
c_{3}B^{2}+c_{1}A\right) b_{\mu }\, ,  \label{eq:D_dot_x}\\
D_{~\mu }^{\nu }b_{\nu } &=&\left( c_{2}-\dfrac{1}{2}c_{1}\right)
\left\langle b,b\right\rangle \dot{x}_{\mu }+\left( c_{1}+\dfrac{
c_{3}\left\langle b,b\right\rangle }{2}\right) Bb_{\mu } \, ,\label{eq:D_b}
\end{eqnarray} 
where 
\begin{align}
\left\langle b,b\right\rangle = a^{\mu\nu}(x)b_\mu b_\nu\,.
\end{align}

In the following sections, we will completely integrate the pseudo-Finsler metrizability conditions \eqref{def:Berwald_conditions}, for an arbitrary 4-dimensional metric-affine structure with vectorial nonmetricity, i.e., for horizontal basis vectors given by~\eqref{eq:CoeffsDnumu}-\eqref{eq:D^mu_nu}.
The only restriction we will impose \textit{a priori} is that the Berwald-Finsler Lagrangian to be determined has an algebraic dependence on the input data  $a_{\mu\nu}, b_{\mu}$ (note: whereas we cannot exclude a differential dependence of $L$ on $a_{\mu\nu}, b_{\mu}$, \textit{i.e.}, on $\Gamma^{\mu}_{~\nu\rho}$, algebraic dependence is what we expect in most cases, as the canonical connection coefficients $\Gamma^{\mu}_{~\nu\rho}$ depend differentiably on $L$, and not vice-versa).

We begin our investigation by studying the Berwald condition for $(\alpha,\beta)$-Finsler structures in the next Section \ref{sec:(alpha,beta)}, as they are a separate branch of the general result and they are easy to handle computationally, before we continue with generalized $(\alpha,\beta)$-Finsler structures in section \ref{sec:generalized_ab}.

\section{Particular case: \texorpdfstring{$(\alpha,\beta)$}{alphabeta}-metrizability}
\label{sec:(alpha,beta)}
Assume that $a$ is a pseudo-Riemannian metric with Levi-Civita connection $\overset{\circ}{\nabla}$, $\nabla = \overset{\circ}{\nabla}+D$ is a torsion-free connection with vectorial nonmetricity as above. Moreover, let $b$ be the one-form on $M$ that defines the vectorial nonmetricity.

By definition, see, e.g., \cite{shen_unicorns}, an \textit{${(\alpha,\beta)}$-metric} is a pseudo-Finsler structure $L$ defined as
   \begin{equation}  \label{def_(alpha,beta)}
       L(x,\dot x)=A \Phi(s), \; \; s:=\frac{B^2}{A},
   \end{equation}
where $\Phi$ is a nontrivial real-valued smooth function of one variable\footnote{the precise domain of definition for $s$ will be typically considered as the maximal one where $\Phi$ is well defined and smooth.} and $A$ and $B$ are defined as in \eqref{eq:AB}.

Note that in \eqref{def_(alpha,beta)}, we use a slightly different notation than typically used in the Finsler literature, where $s$ denotes the ratio $\frac{B}{\sqrt{A}}$; we prefer equation \eqref{def_(alpha,beta)} (also used in {\cite{Fuster:2018djw,Voicu:2023zem}}), since it allows us to identify more directly polynomial expressions in $\dot{x}$, which will be essential in our reasoning below.

Another important remark is that the ansatz \eqref{def_(alpha,beta)} automatically implies the positive 2-homogeneity of $L$ with respect to the vector variable $\dot{x}$; it is hence sufficient to look for nondegenerate solutions of the PDE system \eqref{def:Berwald_conditions} with input data \eqref{eq:D^mu_nu}.

\bigskip

The main result of this section is formulated as follows.

\begin{theorem}\label{thm:metrizability_classification_alpha_beta}
   A connection with vectorial nonmetricity, locally given by
    \begin{align}
    \Gamma^\mu{}_{\nu\rho} = \overset{\circ}{\Gamma}{}^\mu{}_{\nu\rho} + b^\mu a_{\nu \rho} \left(\frac{2 c_2 - c_1}{2} \right) + \frac{c_1}{2} \delta^{\mu}_{\nu} b_{\rho}+\frac{c_1}{2} \delta^{\mu}_{\rho} b_{\nu} +\frac{c_3}{2} b^\mu b_\nu b_\rho
\end{align}
where $c_1,c_2,c_3$ are not all zero and the 1-form $b$ is not absolutely parallel with respect to the Levi-Civita connection $\overset{\circ}{\nabla} $ of the Riemannian metric $a$, is pseudo-Finsler metrizable by a Berwald-type $(\alpha,\beta)$-metric $L=A \Phi$, with $\Phi=\Phi(s)$, if and only if one of the following happens:
\begin{enumerate}
    \item $c_2=c_3=0$ (Weyl connections) and there exists a constant  $\lambda \in \mathbb{R}\setminus \{0,1\}$, such that the defining one-form $b$ satisfies
    \begin{equation} \label{eq:constraints_b_alpha_beta1}
        \overset{\circ}{\nabla}_{\mu} b_{\nu}=\frac{c_1}{2} \left( - \langle b, b \rangle a_{\mu \nu} + \left( \frac{1}{\lambda} +1 \right) b_{\mu} b_{\nu} \right).
    \end{equation}
    In this case, $L$ is given by a power law 
    \begin{equation}
        L=\kappa A s^{\lambda}, \; \; \kappa \in \mathbb{R}^{*}.
    \end{equation}

    \item $c_2=0, c_3 \neq 0$ and there exists $\tau \in \mathbb{R}$, $\tau \neq \frac{c_1}{c_3}$, such that the defining one-form $b$ satisfies
    \begin{equation} \label{eq:constraints_b_alpha_beta2}
        \overset{\circ}{\nabla}_{\mu} b_{\nu}=\frac{c_3}{2} \left(  -\frac{c_1}{c_3} \langle b,b \rangle a_{\mu \nu} + \left( \frac{c_1}{c_3} + \tau + \langle b, b \rangle \right) b_{\mu} b_{\nu} \right).
    \end{equation}
    In this case:
    \begin{enumerate}
        \item[$(i)$] If $c_1 \neq 0, \tau \neq 0$, then $L$ is a generalized $m$-Kropina metric
        \begin{equation}
            L =\kappa A s^{\tfrac{c_1}{\tau c_3}}\left( s+\tau \right) ^{1-\tfrac{ c_1 }{\tau c_3 }},~\ \kappa \in \mathbb{R}^{\ast } \, ,
        \end{equation}
           \item[$(ii)$] If $c_1=0, \tau \neq 0$ (completely symmetric connections), then $L$ is pseudo-Riemannian
        \begin{equation}
            L=\kappa \left( \tau A+B^{2}\right), ~\ \kappa \in \mathbb{R}^{\ast } \,,
        \end{equation}
              \item[$(iii)$] If $c_1 \neq 0, \tau =0$, then $L$ is nondegenerate and of exponential type
        \begin{equation}
            L =\kappa Ase^{-\tfrac{c_1 }{c_3 s}}=\kappa B^2 e^{-\tfrac{c_1 }{c_3 s}}.
        \end{equation}
    \end{enumerate}
\end{enumerate}
\end{theorem}
\noindent Here are some remarks on the above result. 

First, the above result is consistent both with the classification of (positive definite) $(\alpha,\beta)$- Finsler metrics of Berwald type known from \cite{shen_unicorns} and for the one (in Lorentzian signature), of generalized m-Kropina metrics of Berwald type reported in \cite{Heefer:2022sgt, Pfeifer:2019tyy, Fuster:2020upk}.

Second, we note that,  as long as the one-form  $b$  satisfies the constraints of type  \eqref{eq:constraints_b_alpha_beta1}-\eqref{eq:constraints_b_alpha_beta2} (in particular, it is torse-forming \cite{gherici2026torseformingvectorfieldcertain} and closed), Weyl and completely symmetric connections are  $(\alpha,\beta)$-metrizable, while Schrödinger connections are never $(\alpha,\beta)$-metrizable, as can be seen by comparing the results of theorem \ref{thm:metrizability_classification_alpha_beta} with table~\ref{table:constraints}. This problem will be solved by passing to generalized $(\alpha,\beta)$-metrics (see the next section \ref{sec:generalized_ab}).

Finally, the request that the 1-form $b$ should not be absolutely parallel with respect to the Levi-Civita connection of $a$ is justified as, assuming the contrary, we would get that all $(\alpha,\beta)$-metrics $L$ constructed from these would be affinely equivalent to $a$ -- hence to the impossibility of metrizing via $L$ any connection with nonzero distortion (indeed, if $\overset{\circ }{\nabla}b =0$, we would get $\overset{\circ }{\delta }_{\mu }B =0$, which, together with $\overset{\circ }{\delta }_{\mu }A =0$, leads to $\overset{\circ }{\delta }_{\mu }L =0$, meaning that $L$ actually metrizes $\overset{\circ }{\nabla}$).

\bigskip

The proof of Theorem \ref{thm:metrizability_classification_alpha_beta} relies on several lemmas, which we state below.
The first lemma is obtained by using the $(\alpha,\beta)$-metric ansatz in the PDE system  \eqref{def:Berwald_conditions} (for the moment, for an \emph{arbitrary } symmetric affine connection $\nabla$).
\begin{lemma}
\label{lem:metrizability_(alpha,beta)_D} The Finsler $\left( \alpha ,\beta
\right) $-metric $L=A\Phi $, with $\Phi=\Phi(s)$, metrizes the symmetric affine connection $\nabla
=\overset{\circ }{\nabla}+~D \in Conn(M)$ if and only if $\Phi $ and $B$ solve the system
\begin{equation}
\Phi ^{\prime }B\left( A\overset{\circ }{\delta }_{\mu }B-AD_{~\mu }^{\nu
}b_{\nu }+BD_{~\mu }^{\nu }\dot{x}_{\nu }\right) =\Phi AD_{~\mu }^{\nu }\dot{%
x}_{\nu }\, ,~\ \ \forall \mu =0,..,3 \, .\   \label{eq:Phi_polyn}
\end{equation}
\end{lemma}
\begin{proof} 
We note that
\begin{equation}
\overset{\circ }{\delta }_{\mu }s=2\dfrac{B}{A}\overset{\circ }{\delta }
_{\mu }B,~\ \ \dot{\partial}_{\mu} s=2B\dfrac{b_{\mu }A-B\dot{x}_{\mu }}{A^{2}},
\label{eq:derivs_s}
\end{equation}
where, in the first relation, we have used that $\overset{\circ }{\delta }
_{\mu }A=0$. Together with $L=A\Phi \left( s\right) ,$ this leads to
\begin{equation}
\dot{\partial}_{\mu} L=2\dot{x}_{\mu }\Phi +A\Phi ^{\prime }\dot{\partial}_{\mu}s=2\dot{x}
_{\mu }\Phi +2\Phi ^{\prime }\dfrac{Bb_{\mu }A-B^{2}\dot{x}_{\mu }}{A}.
\end{equation}
and subsequently, to
\begin{equation}
\delta _{\mu }L=\overset{\circ }{\delta }_{\mu }L-D_{~\mu }^{\nu }\dot{\partial}_{\nu} L=\dfrac{2}{A}\left[ B\Phi ^{\prime }\left( A\overset{\circ }{\delta }
_{\mu }B-D_{~\mu }^{\nu }\left( b_{\nu }A-B\dot{x}_{\nu }\right) \right)
-A\Phi D_{~\mu }^{\nu }\dot{x}_{\nu }\right] ,
\end{equation}
which immediately gives the claimed result as $\delta_\mu L = 0$ is equivalent to the fact that the Finsler function $L$ metrizes the connection.
\end{proof}
Here are some immediate remarks on the necessary equations \eqref{eq:Phi_polyn}:
\begin{enumerate}
\item If the nonmetricity $Q$ of $\nabla$ is nonzero, then there exists at
least one index $\mu_0 \in \left\{ 0,1,2,3\right\} $ such that both the left and right-hand sides of the $\mu=\mu_0$-th equation \eqref{eq:Phi_polyn} are not identically zero. 

Indeed, assuming the contrary and taking into account that, none of $\Phi, \Phi',A$ and $B$  can vanish identically\footnote{The vanishing of $\Phi'$ would entail that $L = c A$ and thus $L$ would not metrize the connection.}, this implies that, for all $\mu$, we would have $D_{~\mu }^{\nu }\dot{x}_{\nu }=0,~\ \overset{\circ }{\delta }_{\mu
}B=D_{~\mu }^{\nu }b_{\nu }\,$. But then,  using the explicit expression for the coefficients from equation \eqref{eq:CoeffsDnumu} and differentiating the first relation with respect to $\dot{x}^{\rho}$ and $\dot{x}^{\tau}$, together with the definition of the nonmetricity tensor, we would obtain $D_{\rho \mu \tau }+D_{\tau \mu \rho }=Q_{\mu \tau \rho }=0 \,$,  in contradiction with our initial assumption $Q \neq 0$. 
\item In general, the nondegeneracy of the Riemannian metric $a$ implies that the polynomial in $\dot{x}$ expressions $A$ and $B$ cannot have any common factors. This fact will now play an important role to solve \eqref{eq:Phi_polyn}.
\end{enumerate}

Using these remarks, we obtain a second lemma.

\begin{lemma}\label{lem:rho1} 
If the $\left( \alpha ,\beta \right) $-metric function $L=A\Phi$, with $\Phi=\Phi(s)$, metrizes the symmetric affine connection $\nabla =\overset{\circ }{\nabla }+D$ on $M,$ then there exists a
nonvanishing one-form $\rho =\rho _{\mu }(x)dx^{\mu }$ on $M$ and constants $m,n,q\in \mathbb{R}$ with $m^{2}+q^{2}\not=0$ and $%
n^{2}+q^{2}\not=0$ such that, in any local chart we have
\begin{equation}
\left\{ 
\begin{array}{l}
D_{~\mu }^{\nu }\dot{x}_{\nu }=\rho _{\mu }\left( mA+qB^{2}\right) \\ 
\overset{\circ }{\delta }_{\mu }B-D_{~\mu }^{\nu }b_{\nu }=\rho _{\mu
}B\left( n-m\right)%
\end{array}%
\right. .  \label{eq:D_AB}
\end{equation}
\end{lemma}
\begin{proof}
Fix an arbitrary local chart on $M$. By the first remark,  there exists at least one index $\mu = \mu_0 \in \{0,1,2,3\}$ such that the $\mu_0$-th equation \eqref{eq:Phi_polyn} can be rewritten as
\begin{equation}
\dfrac{\Phi ^{\prime }}{\Phi }=\dfrac{AD_{~\mu_0 }^{\nu }\dot{x}_{\nu }}{%
B\left( A\overset{\circ }{\delta }_{\mu_0 }B-AD_{~\mu_0 }^{\nu }b_{\nu
}+BD_{~\mu_0 }^{\nu }\dot{x}_{\nu }\right) } \, .  \label{eq_rho}
\end{equation}
The left-hand side of the above $\mu_0$-th equation is a function of $s=\frac{B^2}{A}$ only. On the other hand, the right-hand side is a ratio of two homogeneous fourth degree polynomials in the vector coordinates $\dot x^{\mu}$; such a ratio can be equal to a function of $s$ only if both the numerator and denominator are linear combinations -- with coefficients depending solely on the coordinates of the point $x$ -- of $A^2,AB^2,B^4$. Since $A$ and $B$ share no common factors, the structure of equation \eqref{eq_rho} forces the numerator to be a linear combination of $A^2$ and $AB^2$ and the denominator to be a linear combination of $AB^2$ and $B^4$.\\ 
Furthermore, any separate $x$-dependence (not encoded in $s$) must be simplified in the right-hand side; that is, it can only appear through a common overall factor $\rho_{\mu_0}=\rho_{\mu_0}(x)$ in both the numerator and the denominator. Hence, there exist constants $m,l,n,q$ such that
\begin{equation} \label{syst_aux_Lemma2}
\left\{ 
\begin{array}{l}
D_{~\mu_0 }^{\nu }\dot{x}_{\nu }=\rho _{\mu_0 }\left( mA+lB^{2}\right) \\ 
A\overset{\circ }{\delta }_{\mu_0 }B-AD_{~\mu_0 }^{\nu }b_{\nu }+BD_{~\mu_0 }^{\nu
}\dot{x}_{\nu }=B\rho _{\mu_0 }\left( nA+qB^{2}\right)%
\end{array}
\right. \, 
\end{equation}
and consequently:
\begin{equation} \label{eq:diffeq_Phi_alpha_beta_rough}
    \frac{\Phi '}{\Phi}=\frac{A(mA+lB^2)}{B^2(nA+qB^2)} \,.
\end{equation}
Observe that, in case when the bracket on the left-hand side of \eqref{eq:Phi_polyn} vanishes, still \eqref{syst_aux_Lemma2} holds, with $\rho_{\mu_0}=0$.

Substituting the expression of $D^{\nu}{}_{\mu_0} \dot x_{\nu}$ from the first equation of \eqref{syst_aux_Lemma2} into the second one, we obtain
\begin{equation}\label{eq:constraint}
A\left( \overset{\circ }{\delta }_{\mu_0 }B-D_{~\mu_0 }^{\nu }b_{\nu }\right)
=\rho _{\mu_0 }B\left[ B^{2}\left( q-l\right) +A\left( n-m\right) \right] \, .
\end{equation}
Since $A$ and $B$ have no common factors, equation ~\eqref{eq:constraint} immediately implies $l = q$, which proves \eqref{eq:D_AB}. Moreover, as $A,B$ are coordinate-invariant and $D^{\nu}_{~\mu_0}$ are components of a globally well-defined tensor field on $M$, it follows immediately that $\rho_{\mu_0}$ must be the components of a globally well defined one-form.
Finally, the pairs $(m,q)$ and $(n,q)$ cannot simultaneously vanish, as this would force the corresponding left-- and right-hand sides of equation \eqref{eq:Phi_polyn} to vanish identically for all $\mu_0$, which is excluded. Therefore, we must have $m^2+q^2 \neq 0$ and $n^2 +q^2 \neq 0$.
\end{proof}
From the proof of the above lemma, it turns out that $\Phi$ is determined by a simple, integrable first-order relation
\begin{equation}
\dfrac{\Phi ^{\prime }}{\Phi }=\dfrac{m+qs}{s(n+qs)} \, .  \label{eq:Phi_AB}
\end{equation}
Yet,  let us recall that our original PDE system, which is equivalent to (\ref{eq:D_AB}), is overdetermined, hence will impose consistency conditions which need to be studied before jumping to integration. To study these consistency conditions, we will pass to the case of vectorial nonmetricity and use the explicit form \eqref{eq:D^mu_nu} of the distortion. This leads to a third lemma.

\begin{lemma}\label{prop:sufficientcondalphabeta}
     If an $(\alpha,\beta)$ metric $L=A \Phi$, where $\Phi=\Phi(s)$, metrizes a connection with vectorial nonmetricity, there exists a nonvanishing one-form $\rho=\rho_{\mu}(x) dx^{\mu}$ on $M$, and constants $m,n,q \in \mathbb{R}$ with $m^2+q^2 \neq 0$ and $n^2+q^2 \neq 0$, such that these are related to the coefficients $c_1,c_2,c_3$ and the one-form $b_{\mu}$ of the vectorial nonmetricity by
\begin{equation} \label{eqs:Lemma3}
c_{2} =0, \; \; c_{3}m=c_{1}q, \; \; \rho _{\mu }q =\frac{c_{3}}{2}b_{\mu }, \; \;  \rho _{\mu }m=\dfrac{c_{1}}{2}
b_{\mu }, ~\ \\
\overset{\circ}{\nabla}_{\mu}b_{\nu} =n\rho _{\mu }b_{\nu }+\dfrac{1}{2}\left( c_{1}+\left\langle
b,b\right\rangle c_{3}\right) b_{\mu }b_{\nu }-\dfrac{1}{2}\left\langle
b,b\right\rangle c_{1}a_{\mu \nu } \, .
\end{equation}
\end{lemma}
The fact that the covariant derivative of $b$ needs to split into a term proportional to $a_{\mu\nu}$ and a term proportional to $b_\mu$ is also known as $b$ being \emph{torse forming}.

\begin{proof}
Assume that the $L=A \Phi \left(\frac{B^2}{A} \right)$  metrizes the connection $\nabla = \overset{\circ}{\nabla}+D$ with vectorial nonmetricity \eqref{eq:vectorialnonmetricity} and substitute the explicit expressions of $D^{\nu}{}_{\mu} \dot x_{\nu}$ and $D^{\nu}{}_{\mu} b_{\nu}$ from equations \eqref{eq:D_dot_x}-\eqref{eq:D_b} into the system \eqref{eq:D_AB}. A direct computation shows that this becomes
\begin{equation}
\left\{ 
\begin{array}{l}
c_{2}B\dot{x}_{\mu }+\dfrac{1}{2}\left( c_{3}B^{2}+c_{1}A\right) b_{\mu
}=\rho _{\mu }\left( mA+qB^{2}\right) \\ 
\overset{\circ }{\delta }_{\mu }B-\left[ \left\langle b,b\right\rangle
\left( c_{2}-\dfrac{1}{2}c_{1}\right) \dot{x}_{\mu }+\left(
c_{1}+\left\langle b,b\right\rangle \dfrac{c_{3}}{2}\right) Bb_{\mu }\right]
=\rho _{\mu }B\left( n-m\right)
\end{array}
\right. .  \label{eq:D_AB_vectorial}
\end{equation}
The first equation of \eqref{eq:D_AB_vectorial} will provide the claimed algebraic constraints; then substituting the obtained constraints into the second equation will lead to the differential constraints on $b_{\nu}$.  
\begin{enumerate}
    \item[\textit{Step 1.}] \textit{ Algebraic constraints:}
    \begin{enumerate}
        \item[$1$.] Contracting the  first equation \eqref{eq:D_AB_vectorial}  with $\dot{x}^{\mu}$ leads to
    \begin{equation}
\left( \dfrac{1}{2}c_{1}+c_{2}\right) AB+\dfrac{c_{3}}{2}B^{3}=\left( \rho
_{\mu }\dot{x}^{\mu }\right) \left( mA+qB^{2}\right) \,.
\end{equation}
Taking, again, into account that $A$ and $B$ (regarded as polynomial expressions in $\dot{x}$) have no common factors, we find
\begin{equation}
\left( \rho _{\mu }\dot{x}^{\mu }\right) q=\dfrac{c_{3}}{2}B,~\ \left( \rho
_{\mu }\dot{x}^{\mu }\right) m=\left( \dfrac{1}{2}c_{1}+c_{2}\right) B.
\end{equation}
Differentiation with respect to $\dot{x}^{\mu }$ then implies
\begin{equation}
\rho _{\mu }q=\dfrac{c_{3}}{2}b_{\mu }, \; \; 
\rho _{\mu }m=\left( \dfrac{1}{2}c_{1}+c_{2}\right) b_{\mu } \,  ,\label{eq:first_constraints_rho}
\end{equation}
which upon contraction with $b^{\mu}$ reveals that
\begin{equation} \label{eq:compareconstraint1}
\left( \rho _{\mu }b^{\mu }\right) q=\dfrac{c_{3}}{2}\left\langle
b,b\right\rangle, \; \;
\left( \rho _{\mu }b^{\mu }\right) m=\left( \dfrac{1}{2}c_{1}+c_{2}\right)
\left\langle b,b\right\rangle \, .
\end{equation}
\item[$2.$] Contracting the first equation of \eqref{eq:D_AB_vectorial} with $b^\mu$  and separating coefficients of $A$ and $B^2$ provides two additional relations
    \begin{equation}\label{eq:compareconstraint2}
            \frac{1}{2} c_1 \langle b, b \rangle = (\rho_\mu b^\mu) m, \qquad
            c_2 + \frac{c_3}{2} \langle b, b \rangle = (\rho_\mu b^\mu) q \,.
        \end{equation}
 \end{enumerate}
Comparing the two sets of constraints  \eqref{eq:compareconstraint1}-\eqref{eq:compareconstraint2} immediately gives the necessary consistency condition
\begin{equation}
c_{2}=0 \, .  \label{eq_c2=0}
\end{equation}
 The latter, substituted into \eqref{eq:first_constraints_rho}, produces the additional relation
\begin{equation}
c_{3}m=c_{1}q \, .  \label{eq:cmq_consistency}
\end{equation}
Summarizing, we have obtained the claimed equalities
\begin{equation}\label{eq:algebraicconstraintsfinal}
  c_2=0, \; \; c_3m=c_1q, \; \; \rho_{\mu} q=\frac{c_3}{2} b_{\mu}, \; \; \rho_{\mu}m=\frac{c_1}{2} b_{\mu} \, .
\end{equation}
    \item[\textit{Step 2.}] \textit{ Differential constraints on $b$.} Plugging the algebraic constraints into the second equation of \eqref{eq:D_AB_vectorial} simplifies it to
\begin{equation}
\overset{\circ}{\delta}_\mu B = n B \rho_\mu + \frac{1}{2} B (c_1 + \langle b,b \rangle c_3) b_\mu - \frac{1}{2} \langle b,b \rangle c_1 \dot{x}_\mu \, ,
\end{equation}
which upon differentiation with respect to $\dot{x}^\nu$ yields the last equation \eqref{eq:D_AB_vectorial}.
\end{enumerate}
\end{proof}

Here is an interesting side remark. Lemma \ref{prop:sufficientcondalphabeta} implies that
\begin{equation}
\rho_{\mu} = \mathrm{const}\cdot b_{\mu} \, ,
\end{equation}
which in turn implies the symmetry  $\overset{\circ}{\nabla}_{\nu} b_{\mu} - \overset{\circ}{\nabla}_{\mu} b_{\nu}=0$. Hence, the one-form $b$ must be closed, 
\begin{equation}
    \mathrm{d}b=0.
\end{equation}

We are now ready to prove the main statement of this section.

\begin{proof} \textit{of Theorem} \ref{thm:metrizability_classification_alpha_beta}. \\
\textit{Necessity}: By Lemma \ref{prop:sufficientcondalphabeta}, a necessary condition for a connection with vectorial nonmetricity to be pseudo-Finsler-metrizable is the existence of a nonvanishing one-form $\rho=\rho_{\mu}(x) dx^{\mu}$ and of the constants $m,n,q \in \mathbb{R}$ with $m^2+q^2 \neq 0$ and $n^2+q^2 \neq 0$, such that equations \eqref{eqs:Lemma3} hold. Noting that in any case, $c_2=0$, we distinguish two cases:
\begin{enumerate}
\item If  $c_3=0$, then, the second equation  of \eqref{eqs:Lemma3}  implies $q=0$. Further, 
from $m^2+q^2 \neq 0$ and $n^2+q^2 \neq 0$, we find that neither $m$, nor $n$ can vanish, that is, we can write:
    \begin{equation}
        \rho_{\mu}=\frac{c_1}{2m} b_{\mu} \, .
    \end{equation}
 Substituting the above value into the last equation of \eqref{eqs:Lemma3}, we get the differential constraint \eqref{eq:constraints_b_alpha_beta1} with
\begin{equation}
\lambda = \frac{m}{n} \neq 0 \; .
\end{equation}
    \item If $c_3\neq 0$, then, the third equation of \eqref{eqs:Lemma3} ensures that $q \neq 0$ (otherwise, we would get that $b$ identically vanishes, which is excluded by hypothesis). Denoting:
\begin{equation}
\Lambda:=\frac{m}{q}=\frac{c_1}{c_3},\qquad \tau:=\frac{n}{q} \,,
\end{equation}
the differential constraint in \eqref{eqs:Lemma3} becomes precisely \eqref{eq:constraints_b_alpha_beta2}.
\end{enumerate}

Moreover, we have shown that the pseudo-Finsler structure $L=A\Phi$, if any, must necessarily be obtained by  integrating the differential equation
\begin{equation}\label{eq:Phi}
    \frac{\Phi'}{\Phi}=\frac{m+qs}{s(n+qs)} \, .
\end{equation}
\textit{Sufficiency:} Assuming that all the above necessary conditions are satisfied, we will prove that a pseudo-Finsler structure metrizing $\nabla$ exists. To this aim, we integrate equation \eqref{eq:Phi}, as follows. 
\begin{enumerate}
     \item If  $c_3=0$, hence $q=0$, using the above notations, our equation reduces to $\frac{\Phi'}{\Phi}= \frac{\lambda}{s} \,$,
    thus leading to the  power-law metric
   \begin{equation}\label{eq:powl}
        \Phi(s)= \kappa s^{\lambda} \, \Rightarrow L=\kappa A s^{\lambda}, \; \; \kappa \in \mathbb{R}^{*} \, .
    \end{equation}
    \item  If $c_3\neq 0$, hence, $q \neq 0$, then equation \eqref{eq:Phi} takes the form
\begin{equation}
\dfrac{\Phi ^{\prime }}{\Phi }=\dfrac{\Lambda +s}{s(\tau +s)} \, .
\end{equation}
Integrating gives three distinct families of solutions:
\begin{enumerate}
   \item[$(i)$] If $\Lambda,\tau \neq 0$, or equivalently, $c_1,\tau \neq 0$, then $\Phi$ is of generalized $m$-Kropina type ($\Lambda, \tau \neq 0$)
\begin{equation}\label{eq:Phisol_AB_2}
\Phi =\kappa s^{\tfrac{\Lambda }{\tau }}\left( s+\tau \right) ^{1-\tfrac{%
\Lambda }{\tau }} \Rightarrow  L=\kappa A s^{\tfrac{\Lambda}{\tau}}(s+\tau)^{1-\tfrac{\Lambda}{\tau}}, \; \;  \kappa \in \mathbb{R}^{\ast } \,.
\end{equation}
\item[$(ii)$] If $ \Lambda=0, \tau \neq 0$, or equivalently, $c_1=0, \tau \neq 0$, then $\Phi$ is Riemannian
\begin{equation}\label{eq:Phisol_AB_3}
\Phi =\kappa \left( \tau +s\right) \Rightarrow L=\kappa \left( \tau
A+B^{2}\right), \; \; \kappa \in \mathbb{R}^{\ast } \, .
\end{equation}
\item[$(iii)$] If $\Lambda \neq 0, \tau =0$, or equivalently, $c_1 \neq 0, \tau =0$, then $\Phi$ is of exponential type
\begin{equation}\label{eq:Phisol_AB_4}
\Phi =\kappa se^{-\tfrac{\Lambda }{s}} \Rightarrow L=\kappa B^2 e^{-\tfrac{\Lambda }{s}}, \; \; \kappa \in \mathbb{R}^{\ast}.
\end{equation}
\end{enumerate}
\end{enumerate} 
Direct computation confirms that all four solutions satisfy the metrizability conditions exactly. 
To check the nondegeneracy of the obtained solutions, we calculate $\det(g)$ in each case. Using  the identity (see \cite{Voicu:2023zem})

\begin{equation}\label{eq:detg}
\det(g) = \Phi^2 (\Phi - s\Phi') \frac{\partial}{\partial s}\left(  (s-\langle b,b \rangle)\frac{(\Phi - s\Phi')^2}{\Phi}  \right) \det(a)\, ,
\end{equation}
we find that:
\begin{itemize}
    \item[$\triangleright$]    For the power law metric, one needs to impose \eqref{eq:powl}: $\lambda \neq 1$;
    \item[$\triangleright$] The generalized $m$-Kropina type metric \eqref{eq:Phisol_AB_2} is nondegenerate whenever $\Lambda-\tau \neq 0$, that is $\tau \neq \frac{c_1}{c_3}$. 
    \item[$\triangleright$] The shifted Riemannian metric is always nondegenerate (indeed, degeneracy would imply $\tau + \langle b, b\rangle = 0,$ which, taking into account that, in this case, $c_1 = 0$, would lead to an identically vanishing Levi-Civita covariant derivative of $b$, which is forbidden by hypothesis).
    \item[$\triangleright$] The exponential type metric \eqref{eq:Phisol_AB_4} is always nondegenerate.
\end{itemize}
This completes the proof, and thus we found all possible solutions, i.e., the only $(\alpha,\beta)$-metrics that metrize connections with vectorial nonmetricity.
\end{proof}

Having shown that $(\alpha,\beta)$-Finsler structures can metrize some of the connections with vectorial nonmetricity, immediately the question arises, if we can do better and metrize all of these connections by more general Finsler structures. Therefore we investigate this question in the following.

\section{Generalized \texorpdfstring{$(\alpha,\beta)$}{generalizedalphabeta}-metrizability}
\label{sec:generalized_ab}
We will now find the most general pseudo-Finsler Lagrangians which metrize a given symmetric affine connection $\nabla$ with vectorial nonmetricity and which depend algebraically on the pseudo-Riemannian metric $a$ and on the one-form~$b$. The first subsection states the main result, which is proved in the following subsections.

\subsection{Main statement}
As the only independent scalar invariants one can algebraically construct from $a_{\mu\nu}$, $b_{\mu}$  and $\dot{x}^{\mu}$ alone are $ \langle b, b \rangle:=a^{\mu\nu}b_{\mu}b_{\nu},\, A$ and $B$,  the Finsler Lagrangians we are looking for must depend on these three quantities only. To disentangle the dependence of $L$ on the pseudo-norm $\langle b, b \rangle$ and on the direction of $b$, we will introduce the notations: 
\begin{equation} \label{defs:|b|_U}
   |b|=\sqrt{| \langle b, b \rangle |}, \; \; \; \; b_{\mu}=|b| u_{\mu}, \; \; \; \; U=u_{\mu} \dot x^{\mu}.
\end{equation}
Thus, the 2-homogeneity condition on $L$ entails: 
\begin{equation} \label{def:generalized_ab_metric}
    L=A \Phi \left( |b|, p \right),
\end{equation}
where:
\begin{equation}
   p=\frac{U^2}{A}.
\end{equation}
We note that, according to the above definition, $u=u_{\mu}dx^{\mu}$ is a normalized one-form on $M$, meaning that
\begin{equation}
u_{\mu}u^{\mu}=\epsilon, \; \; \text{where} \; \; \epsilon=\pm 1.
\end{equation}
In particular, this implies: $b_{\mu}b^{\mu}= \langle b, b\rangle=|b|^2 \epsilon$.

\bigskip

\textbf{Remark}.    We can safely assume in the following that $|b| \neq 0$  and, apart from possible isolated points,
    \begin{equation}\label{eq:phi_b}
        \Phi_{|b|}' \equiv \partial_{|b|}\Phi \neq 0.
    \end{equation}
Indeed, if any of these conditions fails, the generalized $(\alpha,\beta)$-metric \eqref{def:generalized_ab_metric} reduces to a standard $(\alpha,\beta)$-metric -- the case which has already been separately discussed in the previous Section \ref{sec:(alpha,beta)}.

In the following, we will find the consistency conditions and integrate the overdetermined PDE system \eqref{def:Berwald_conditions}, using, this time, \eqref{def:generalized_ab_metric} as our ansatz.

\bigskip

Here is our main result.
\begin{theorem} \label{thm:metrizability_general_ab}
  A connection $\nabla = \overset{\circ}{\nabla}+D$ with nonzero vectorial nonmetricity is pseudo-Finsler-metrizable by a generalized $(\alpha,\beta)$-metric $L=A\Phi(|b|,p)$ if and only if the following conditions are simultaneously satisfied:
  \begin{enumerate}
    \item $c_3=0$, or $c_1=c_2=0$.
    \item $c_3$ and $c_1$ do not simultaneously vanish.
    \item There hold the equalities
\begin{equation}\label{eq:db}
d|b| = \lambda u, \quad \overset{\circ}{\nabla}_\mu u_\nu=\tau \left(  a_{\mu\nu} - \epsilon u_\mu u_\nu \right), 
\end{equation}
where $\lambda = \lambda(|b|)$ is an arbitrary, nowhere zero smooth function and $\tau = \tau(|b|)$ is defined in terms of $\lambda$ as
\begin{equation}\label{eq:tau_general}
    \tau=\frac{ c_1 |b|}{ c_1 C_1 e^{\left(c_1 +2c_2 \right)\rho(|b|)} - 2 \epsilon},   \qquad  \text{with} \qquad     \rho(|b|)=\int \frac{|b|}{\lambda(|b|)} d|b|
\end{equation}
and  $C_1 \in  \mathbb{R}$ is an arbitrary constant satisfying $c_2(c_1 - 2c_2)C_1 =0$.

\end{enumerate}

If the above conditions are satisfied, then the solution is explicitly constructed as:
\begin{enumerate}
    \item[$(i)$] If $c_3 \neq 0, c_1=c_2=0$ (completely symmetric connections),  then 
        \begin{equation}\Phi(|b|,p)=\frac{p}{\epsilon}\exp \left({c_3 \epsilon \int \frac{|b|^3}{\lambda(|b|)} d|b|} + F \left( \frac{e^{-c_3 \epsilon \int \frac{|b|^3}{\lambda(|b|)} d|b|}(\epsilon -p) }{p \epsilon}\right) \right), \; \; \text{where} \; \; F \; \; \text{is a free function of one variable}\,.
        \end{equation}
    \item[$(ii)$] If $c_3 = 0,c_1 \neq 0$, then
        \begin{enumerate}
            \item For $c_2=0$ (Weyl type connections):
            \begin{equation}
                \Phi(|b|,p)=\exp \left({\left(\frac{\epsilon}{c_1}-\frac{C_1}{2} e^{c_1 \rho(|b|)}  \right) \left(c_1\epsilon p +\frac{C_2-2c_1 \rho(|b|)}{C_1 e^{c_1 \rho(|b|)}-2\frac{\epsilon}{c_1}} \right)} \right).
            \end{equation}
            \item For $c_2 \neq 0, c_1 - 2c_2 \neq 0,$ (including Schr\"odinger symmetric connections):
            \begin{equation}
                \Phi(|b|,p)=\exp \left({\frac{\epsilon}{c_1} \left\{   c_2 p^2+ \epsilon(c_1-2c_2)p+ \left(C_3 + \epsilon c_1^2 \rho(|b|) \right) \right \}} \right),
            \end{equation}
            \item For $c_2 \neq 0, c_1-2c_2 =0$:
            \begin{equation}
                \Phi(|b|,p)=\exp \left({\left(\frac{\epsilon}{2}-\frac{1}{2} C_1 c_2 e^{4 \rho(|b|) c_2}\right) p^2 + \left(2 c_2 \rho(|b|) - \frac{1}{2} C_4 \right)} \right),
            \end{equation}
        \end{enumerate} 
        where, in the above, $C_1,C_2,C_3,C_4$ are arbitrary constants.
\end{enumerate}
\end{theorem}

The proof of the above statement will be made over the next two subsections, as follows:
\begin{enumerate}
\item In the first subsection, we rewrite the Berwald metrizability PDE system \eqref{def:Berwald_conditions} using the ansatz \eqref{def:generalized_ab_metric} and use the contractions of the newly obtained PDEs with $u$ and $\dot{x}$ to deduce the constraints upon the derivatives of $\langle b,b \rangle$ and $u_\mu$. This is done in four lemmas in Section~\ref{ssec:somelem}.
\item Then, using the obtained expressions, we proceed to the integration of the metrizability conditions in Section~\ref{ssec:prfmain}. 
\end{enumerate}

Before starting the proof, it is worth noting that extending our search to generalized $(\alpha,\beta)$-metrics, we have completely eliminated the constraint $c_2=0$; in particular, it turns out that, under appropriate conditions upon the defining one-form, Schrödinger connections are Finsler metrizable. 

\subsection{Some lemmas on necessary conditions  }\label{ssec:somelem}
The first step will be to rewrite the Berwald metrizability conditions \eqref{def:Berwald_conditions} for our ansatz \eqref{def:generalized_ab_metric}.  We thus get an analogue of Lemma \ref{lem:metrizability_(alpha,beta)_D}.

\begin{lemma}\label{lem:metrizability_generalized_(alpha,beta)_D}
A torsion-free connection $\nabla $ on $M$ is metrizable by a generalized  $\left(
\alpha ,\beta \right) $-metric if and only if its distortion $D$ satisfies
\begin{equation}
\dfrac{1}{2}A^{2}\Phi_{\left\vert b\right\vert }^{\prime} \partial_{\mu} \left\vert b\right\vert 
+ \Phi _{p}^{\prime }U\left( A\overset{\circ }{\delta }_{\mu }U-A\left( D^{\nu}{}_{\mu}u_{\nu }\right) +U\left( D^{\nu}{}_{\mu}\dot{x}_{\nu }\right) \right) 
=A\left( D^{\nu}{}_{\mu} \dot{x}_{\nu }\right) \Phi \,, \qquad \forall \mu=0,\dots,3 \,,
\label{eq:Lemma_general_(alpha,beta)_D}
\end{equation}
where $\Phi _{p}^{\prime }\equiv \partial_p \Phi$ and $\Phi _{|b|}^{\prime }\equiv \partial_{|b|}\Phi$.
\end{lemma}
\begin{proof}
The statement follows from a straightforward computation, substituting $\delta_\mu = \overset{\circ}{\delta}_\mu-D^{\nu}{}_{\mu} \dot \partial_{\nu}$, together with the identities
\begin{equation}\label{eq:deltapdp}
\overset{\circ }{\delta }_{\mu }p=2\dfrac{U}{A}\overset{\circ }{\delta }_{\mu }U = 2 \frac{p}{U}\overset{\circ }{\delta }_{\mu }U,~\quad 
\dot{\partial}_{\mu} p=2U\dfrac{u_{\mu }A-U\dot{x}_{\mu }}{A^{2}},
\end{equation}
into the Berwald metrizability conditions $\delta_\mu L =0$.
\end{proof}

In particular, for a connection with vectorial nonmetricity, the distortion components take the form
    \begin{equation}
D_{~\nu \rho }^{\mu }=\left\vert b\right\vert \left[ \dfrac{1}{2}\left(
2c_{2}-c_{1}\right) u^{\mu }a_{\nu \rho }+\dfrac{1}{2}c_{1}u_{\nu }\delta
_{\rho }^{\mu }+\dfrac{1}{2}c_{1}u_{\rho }\delta _{\nu }^{\mu }+\dfrac{1}{2}
c_{3}\left\vert b\right\vert ^{2}u^{\mu }u_{\nu }u_{\rho }\right] \, ;
\label{eq:distortion_b_u}
\end{equation}
this leads, via $D^{\mu}_{~\nu} = D^{\mu}_{~\nu\rho}\dot{x}^{\rho}$, to 
\begin{eqnarray}
D_{~\nu }^{\mu }\dot{x}_{\mu } &=&\left\vert b\right\vert \left[ Uc_{2}\dot{x
}_{\nu }+\dfrac{1}{2}\left( c_{3}\left\vert b\right\vert
^{2}U^{2}+Ac_{1}\right) u_{\nu }\right] ,  \label{D_x_1} \\
D_{~\nu }^{\mu }u_{\mu } &=&\left\vert b\right\vert \left[ \dfrac{1}{2}
\left( 2c_{2}-c_{1}\right) \epsilon \dot{x}_{\nu }+\left( c_{1}+\dfrac{1}{2}
c_{3}\epsilon \left\vert b\right\vert ^{2}\right) Uu_{\nu }\right] .
\label{D_u_1}
\end{eqnarray}
The next step is to contract the obtained metrizability conditions \eqref{eq:Lemma_general_(alpha,beta)_D} in turn with the components of $\dot{x}$ and $u$ to obtain constraints on the derivatives $\partial_{\mu} \left\vert b\right\vert$ and $\overset{\circ }{\delta }_{\mu }U$.  This will lead to further lemmas.

\begin{lemma}\label{lem:U_prime_u_zero}
  If the connection $\nabla$ with vectorial nonmetricity is metrizable by a generalized  $\left(
\alpha ,\beta \right) $-metric, then, in any local chart:
\begin{enumerate}
\item  $   u^{\mu} \overset{\circ}{\delta}_{\mu} U=0 \,.$
\item The expression $u^{\mu} \partial_{\mu} |b|$ is a function of $|b|$ alone.
\end{enumerate}

\end{lemma}
\begin{proof}
 
Contracting \eqref{eq:Lemma_general_(alpha,beta)_D} with $u^{\mu}$, we obtain:
\begin{equation}
A^{2}\dfrac{1}{2}\Phi _{\left\vert b\right\vert }^{\prime }\left(u^{\mu} \partial_{\mu} \left\vert
b\right\vert\right) +\Phi _{p}^{\prime }U\left( A\left(
u^{\mu }\overset{\circ }{\delta }_{\mu }U\right) -A\left( D^{\nu}{}_{\mu
}u_{\nu }u^{\mu }\right) +U\left( D^{\nu}{}_{\mu}\dot{x}_{\nu }u^{\mu
}\right) \right) =A\left( D^{\nu}{}_{\mu}\dot{x}_{\nu }u^{\mu }\right) \Phi .
\label{eq:contracted_u_1}
\end{equation} 
Since $\Phi'_{|b|} \neq 0$, we can divide equation \eqref{eq:contracted_u_1} by $\frac{1}{2} A^2 \Phi'_{|b|}$ and recast it (after substituting the above expressions for the contracted distortion) as
  
\begin{equation}\label{eq:Lemma_F_G}
     u^{\mu} \partial_{\mu} |b|+ \mathcal{F}(|b|,p) \frac{\left( u^\mu \overset{\circ}{\delta}_{\mu} U \right)}{U}=\mathcal{G}(|b|,p)\, ,
\end{equation}
where
\begin{align}
    \mathcal{F}(|b|,p)&:= \frac{2 p \Phi'_{p}}{\Phi'_{|b|}},\\
    \mathcal{G}(|b|,p)&:=\frac{|b|}{\Phi'_{|b|}} \left\{\Phi \left( 2p c_2 + \epsilon c_1  + \epsilon c_3 p |b|^2 \right) +p \Phi'_{p} \left [\epsilon \left(2 c_2 + c_1 + \epsilon c_3  |b|^2 \right) - \left(2pc_2 + \epsilon c_1 + \epsilon c_3 p |b|^2 \right) \right] \right\} \, .
\end{align}
As $\mathcal{F}(|b|,p)$ is non-vanishing, we can use \eqref{eq:Lemma_F_G} to equate
\begin{align}\label{eq:udU/U}
    \frac{\left( u^\mu \overset{\circ}{\delta}_{\mu} U \right)}{U} = \frac{\mathcal{G}(|b|,p) - u^{\mu} \partial_{\mu} |b|}{\mathcal{F}(|b|,p)}\,;
\end{align}
this is an equality where the ratio on the right-hand side just depends on $\dot x$ only through $p = \tfrac{U^2}{A}$, which is the irreducible ratio of $2$-homogeneous polynomials in $\dot x$, whereas its left-hand side $\tfrac{\left( u^\mu \overset{\circ}{\delta}_{\mu} U \right)}{U}$ is a ratio of $1$-homogeneous polynomials in $\dot{x}$. The only way that \eqref{eq:udU/U} can have non-trivial solutions is if the left-hand side only depends on $x^\mu$, i.e.\
\begin{equation}
    u^{\mu} \overset{\circ}{\delta}_{\mu} U= f U
\end{equation}
for some smooth function $f$ depending on $x^{\mu}$  only. Differentiating this relation with respect to $\dot{x}^{\nu}$, we find
\begin{equation}\label{eq:Lemma_f}
   u^{\mu} \overset{\circ}{\nabla}_{\mu} u_{\nu}= f u_{\nu}.
\end{equation}
On the other hand,  the quantity $u^{\nu} u_{\nu}:=\epsilon$ is either $+1$ or $-1$, so Levi-Civita differentiation yields $ u^{\nu} \overset{\circ}{\nabla}_{\mu} u_{\nu}=0$. Therefore, contracting the previous equation with $u^{\nu}$ immediately implies $f=0$, that is,
\begin{equation}
    u^{\mu} \overset{\circ}{\delta}_{\mu} U=0.
\end{equation}
 The second claim follows directly by substituting the above equality into  \eqref{eq:Lemma_F_G} and taking into account that $ u^{\mu} \partial_{\mu} |b|$ does not depend on $\dot{x}$.
\end{proof}


\begin{lemma} \label{lem:partial_derivs_|b|}
    If a connection with vectorial nonmetricity is metrizable by a generalized $(\alpha,\beta)$-metric, then, corresponding to each chart domain, there exists a function $\lambda=\lambda (|b|)$, such that
    \begin{equation} \label{eq:partial_derivs_|b|}
        \partial_{\mu}|b|= \lambda u_{\mu}.
    \end{equation}
\end{lemma}
\begin{proof}

It is more convenient to work, in the following, in terms of
\begin{equation}\label{def:Psi}
\Psi=\ln \Phi,
\end{equation}
such that the metrizability conditions read 
\begin{equation}\label{eq:general_(alpha,beta)_D_with_Psi}
    \frac{1}{2} A^2 \Psi'_{|b|} \partial_{\mu} |b|+\Psi '_{p} U \left(A \overset{\circ}{\nabla}_{\mu} U- A \left( D^{\nu}{}_{\mu} u_{\nu} \right) +U \left( D^{\nu}{}_{\mu} \dot x_{\nu} \right) \right)=A \left( D^{\nu}{}_{\mu} \dot x_{\nu} \right).
\end{equation}
Contracting these equations with $\dot{x}^{\mu}$ , this becomes:
\begin{equation}\label{eq:metrizability_generalizedalphabeta}
    \frac{1}{2} A^2 \Psi'_{|b|} \left(  \dot x^{\mu} \partial_{\mu}|b| \right)+\Psi'_{p} U \left(A \dot{x}^{\mu} \overset{\circ}{\nabla}_{\mu} U  -A \left( D^{\nu}{}_{\mu} u_{\nu} \dot x^{\mu} \right) + U \left(D^{\nu}{}_{\mu} \dot x_{\nu} \dot x^{\mu} \right) \right)-A \left(D^{\nu}{}_{\mu} \dot x_{\nu} \dot x^{\mu} \right)=0.
\end{equation}
Then, substitution of the contracted distortion
\begin{equation}\label{eq:generalizedalphabetacontractions}
\begin{aligned}
    D^{\nu}{}_{\mu} \dot x^{\mu} \dot x_{\nu}
    &=|b| \left(\frac{1}{2} c_1 + c_2 \right) AU + \frac{1}{2} c_3 |b|^3 U^3\\
    D^{\nu}{}_{\mu} u_{\nu} \dot x^{\mu}
    &=\frac{1}{2} |b| \epsilon (2c_2 -c_1) A + |b| \left(\frac{1}{2} \epsilon c_3 |b|^2 + c_1 \right)U^2
\end{aligned}
\end{equation}
and taking into account the assumption  $\Psi'_{|b|} \neq 0$ made in the beginning of this section, equation \eqref{eq:metrizability_generalizedalphabeta} can be recast as
\begin{equation}\label{eq:contractedxmu_generalized_case}
    \dot x^{\mu} \partial_{\mu} |b| =\frac{\mathcal{E}}{U} \dot{x}^{\mu} \overset{\circ}{\nabla}_{\mu} U + \mathcal{W} U,
\end{equation}
where
\begin{equation}\label{eq:formofEandW}
    \mathcal{E}=-2p \frac{\Psi'_{p}}{\Psi'_{|b|}}, \; \; \mathcal{W}=\left(-\frac{|b|^3}{\Psi'_{|b|}} \Psi'_{p} c_3 \right) p^2 + \frac{|b|}{\Psi'_{|b|}} \left(\Psi'_{p} c_1 - 2\Psi'_{p} c_2 + |b|^2 c_3 + |b|^2 \epsilon \Psi'_{p} c_3 \right)p + \frac{|b|}{\Psi'_{|b|}} \left(c_1+ 2c_2-\epsilon \Psi'_{p}c_1 + 2 \epsilon \Psi'_{p} c_2 \right).
\end{equation}
are smooth functions of $|b|$ and $p$.

Differentiation of  \eqref{eq:contractedxmu_generalized_case} with respect to $\dot{x}^{\nu },$  followed by contraction with $u^{\nu }$ then turns it into:
\begin{equation}
u^{\nu }\partial _{\nu }\left\vert b\right\vert =\mathcal{E}_{p}^{\prime
}\left( u^{\nu }\dot{\partial}_{\nu }p\right) \dfrac{\dot{x}^{\mu }\overset{%
\circ }{\nabla }_{\mu }U}{U}+\mathcal{E}\left( \dfrac{u^{\nu }\overset{\circ 
}{\nabla }_{\nu }U+\dot{x}^{\mu }u^{\nu }\overset{\circ }{\nabla }_{\mu
}u_{\nu }-\epsilon \dot{x}^{\mu }\overset{\circ }{\nabla }_{\mu }U}{U^{2}}%
\right) +\mathcal{W}_{p}^{\prime }\left( u^{\nu }\dot{\partial}_{\nu
}p\right) U+\epsilon \mathcal{W}
\end{equation}

Taking into account point 1. of  Lemma~\ref{lem:U_prime_u_zero}, together with $u^{\nu }\overset{%
\circ }{\nabla }_{\mu }u_{\nu }=0,$ the first two terms in the large bracket vanish, that is:
\begin{equation}
u^{\nu }\partial _{\nu }\left\vert b\right\vert =\left( \dfrac{\mathcal{E}%
_{p}^{\prime }\left( u^{\nu }\dot{\partial}_{\nu }p\right) }{U}-\dfrac{%
\epsilon \mathcal{E}}{U^{2}}\right) \dot{x}^{\mu }\overset{\circ }{\nabla }%
_{\mu }U+\mathcal{W}_{p}^{\prime }\left( u^{\nu }\dot{\partial}_{\nu
}p\right) U+\epsilon \mathcal{W}.
\end{equation}
Lemma \ref{lem:U_prime_u_zero} point 2. then tells us that the left-hand side $%
u^{\nu }\partial _{\nu }\left\vert b\right\vert $ is a function of $%
\left\vert b\right\vert $ alone; since, taking into account \eqref{eq:deltapdp}, we find that $u^{\nu }\dot{\partial}_{\nu }p=\dfrac{%
2U\left( \epsilon A-U^{2}\right) }{A^{2}}$ depends on $A$ and $U$ only, it turns out that $\dot{x}^{\mu }\overset{%
\circ }{\nabla }_{\mu }U$ can only depend on $|b|$, $A$ and $U$.

But, on the other hand, this is a homogeneous second degree polynomial expression in the coordinates of $\dot{x}$. It turns out that the only such possibility is
 \begin{equation}
        \dot x^{\mu} \overset{\circ}{\nabla}_{\mu} U= k_1 (|b|)A + k_2(|b|) U^2,
    \end{equation}
where the coefficients $k_1$ and $k_2$ are smooth functions of $|b|$. Differentiating twice with respect to $\dot{x}$ then leads to
    \begin{equation}\label{eq:contractherefooru}
        \overset{\circ}{\nabla}_{\nu} u_{\mu}+\overset{\circ}{\nabla}_{\mu} u_{\nu}=2(k_1 a_{\mu \nu}+ k_2 u_{\mu} u_{\nu}).
    \end{equation}

    Contractions of the above equation with $u^{\mu}$ and $\dot{x}^{\mu}$ reveal a relation between the coefficients $k_1$ and $k_2$, as follows. First, using $u^{\mu} \overset{\circ}{\delta}_{\mu} U=0$, the contraction of equation \eqref{eq:contractherefooru} with $u^{\mu}$ gives
    \begin{equation}
       u^{\mu} \overset{\circ}{\nabla}_{\nu} u_{\mu} =2( k_1 u_{\nu} + k_2 \epsilon u_{\nu} ).
    \end{equation}
    Then, taking into account that $u^{\mu} \overset{\circ}{\nabla}_{\nu} u_{\mu}=\frac{1}{2} \overset{\circ}{\nabla}_{\nu} \left( u^\mu u_\mu\right)=\frac{1}{2} \overset{\circ}{\nabla}_{\nu} \epsilon=0$,  and since $u_{\nu} \neq 0$, we get
    \begin{equation}
        k_1=-k_2 \epsilon.
    \end{equation}
    Therefore, we can characterize $\dot x^{\mu} \overset{\circ}{\nabla}_{\mu} U$ by a single free function
    \begin{equation}
        \dot x^{\mu} \overset{\circ}{\nabla}_{\mu} U=k(|b|)(U^2-\epsilon A).
    \end{equation}
    Substituting the obtained expression for $\dot x^{\mu} \overset{\circ}{\nabla}_{\mu} U$ in equation \eqref{eq:contractedxmu_generalized_case}, we find that the derivative 
    \begin{equation}
      \dot x^{\mu} \partial_{\mu} |b| = k(|b|) \mathcal{E} \left( U - \epsilon \frac{A}{U} \right) + \mathcal{W}U
    \end{equation}
    can be again completely expressed in terms of $|b|$, $U$ and $A$. On the other hand, the left-hand side is obviously linear in $\dot{x}$, which means the only possibility for the equality to be satisfied is that there exists some $\lambda=\lambda(|b|)$ such that
    \begin{equation}
        \dot x^{\mu} \partial_{\mu} |b| =\lambda(|b|)U;
    \end{equation}
    by $\dot x^{\rho}$ differentiation, we finally obtain
\begin{equation}\label{eq:contracttoreveal_lambda_phi}
        \partial_{\rho} |b| =\lambda(|b|) u_{\rho},
    \end{equation}
which completes the proof of the lemma.
\end{proof}


The above result allows us to finally calculate the Levi-Civita covariant derivative of $u$, as follows.
\begin{lemma} \label{lem:.covariant_derivs_u}
If $\nabla$ is metrizable by a generalized $(\alpha,\beta)$-metric $L=A\Phi(|b|,p)$, then, corresponding to each chart domain, there exists a smooth function $\tau = \tau (|b|)$ such that
    \begin{equation}
        \overset{\circ}{\nabla}_{\mu} u_{\nu}=\tau \left( a_{\mu \nu} - \epsilon u_{\mu} u_{\nu} \right).
    \end{equation}
    
\end{lemma}

\begin{proof}
    Substituting the above found expressions for $\partial_{\mu}|b|$ together with  \eqref{D_x_1}, \eqref{D_u_1}, into the full system \eqref{eq:general_(alpha,beta)_D_with_Psi} turns it, after a direct computation, into:
\begin{eqnarray}
-AU\Psi _{p}^{\prime }\overset{\circ }{\delta }_{\mu }U &=&\left[ \left( 
\dfrac{1}{2}\lambda \Psi _{\left\vert b\right\vert }^{\prime }-\dfrac{1}{2}%
\left\vert b\right\vert c_{1}\right) A^{2}+\left( -\dfrac{1}{2}\left\vert
b\right\vert ^{3}c_{3}-\dfrac{1}{2}\left\vert b\right\vert \Psi _{p}^{\prime
}c_{1}-\dfrac{1}{2}\left\vert b\right\vert ^{3}\epsilon \Psi _{p}^{\prime
}c_{3}\right) AU^{2}+\dfrac{1}{2}\left\vert b\right\vert ^{3}\Psi
_{p}^{\prime }c_{3}U^{4}\right] u_{\mu } \\
&&+\left[ \left( \dfrac{1}{2}\left\vert b\right\vert \epsilon \Psi
_{p}^{\prime }c_{1}-\left\vert b\right\vert c_{2}-\left\vert b\right\vert
\epsilon \Psi _{p}^{\prime }c_{2}\right) AU+\left\vert b\right\vert \Psi
_{p}^{\prime }c_{2}U^{3}\right] \dot{x}_{\mu }.
\end{eqnarray}
Dividing by $-AU\Psi _{p}^{\prime }$ and substituting $A=\dfrac{U^{2}}{p}$, this can be recast as
\begin{equation}\label{eq:U_mu_S_T}
\overset{\circ }{\delta }_{\mu }U=\mathcal{S}Uu_{\mu }+\mathcal{T}\dot{x}_{\mu },
\end{equation}
where the expressions
\begin{eqnarray}\label{eq:expressions_Of_S_and_T}
\mathcal{S} &:&=-\dfrac{1}{2\Psi _{p}^{\prime }}\left[ \dfrac{\lambda }{p}
\Psi _{\left\vert b\right\vert }^{\prime }+\left\vert b\right\vert \left(
-c_{1}+\left\vert b\right\vert ^{2}c_{3}\left( p-\epsilon \right) \right)
\Psi _{p}^{\prime }-\left\vert b\right\vert \left( \left\vert b\right\vert
^{2}c_{3}+\dfrac{c_{1}}{p}\right) \right] ,  \label{def_S} \\
\mathcal{T} &:&=-\dfrac{1}{2\Psi _{p}^{\prime }}\left\vert b\right\vert
\left( 2pc_{2}+\epsilon c_{1}-2\epsilon c_{2}\right)   \label{def:T}
\end{eqnarray}
depend on $\left\vert b\right\vert $ and $p$ only. \\
Taking into account the relation $u^{\mu} \overset{\circ}{\delta}_{\mu} U=0$, the latter reveals, after contracting with $u^{\mu}$ and dividing by $U$, that $\mathcal{T}$ and $\mathcal{S}$ are related as
\begin{equation}\label{eq:relation_between_T_and_S}
    \mathcal{T}=-\mathcal{S} \epsilon \iff \mathcal{S}=-\epsilon \mathcal{T}.
\end{equation}
That is, the four equations \eqref{eq:U_mu_S_T} now read
\begin{equation}\label{eq:x-Uu}
    \overset{\circ}{\delta}_{\mu} U = \mathcal{T} \left(\dot x_{\mu} - \epsilon U u_{\mu} \right),
\end{equation}
where, in principle, $\mathcal{T}$ could depend on both $|b|$ and $p$. 
We will show that, actually, $\mathcal{T}$ can only depend on $|b|$. To this aim, we note that there exists at least one $\mu = \mu_0 \in \{0,1,2,3\}$, such that the factor $\dot{x}_{\mu_0}- \epsilon U u_{\mu_0}$ does not vanish; indeed, assuming the contrary, we would get by contraction of \eqref{eq:x-Uu} with $\dot{x}^{\mu}$ that $A-\epsilon U^2= 0$, which would imply that $A=\epsilon U^2$ is degenerate.
Consequently, there exists at least one index $\mu = \mu_0 \in \{0,1,2,3\}$ for which we can write
\begin{equation}\label{eq:Texpressedfrommetrizability}
    \mathcal{T}=\frac{\overset{\circ}{\delta}_{\mu_0}U}{ \dot x_{\mu_0} - \epsilon U u_{\mu_0}}.
\end{equation}
In the above equality, the left-hand side is a function of $|b|$ and $p$, its $\dot x$ dependence is completely encoded in the 2-homogeneous in $\dot{x}$, relatively prime polynomial expressions $A$ and $U^2$. On the other hand, the right-hand side is a ratio of linear expressions in $\dot x$. This can only be achieved if $  \mathcal{T}=\tau(|b|)$ only.\\
 Substituting this into \eqref{eq:Texpressedfrommetrizability} and differentiating with respect to $\dot{x}^\nu$, we obtain the claim of the lemma.
\end{proof}

Using the above lemmas, we are now able to prove Theorem \ref{thm:metrizability_general_ab}.

\subsection{Proof of the main statement  }\label{ssec:prfmain}
We start with a remark on the necessity of conditions 1.-3. in the statement of Theorem \ref{thm:metrizability_general_ab}.

\begin{enumerate}
\item First, we note that the constraints on the derivatives $\partial_{\mu} |b|$ and $\overset{\circ}{\nabla} u$ obtained in Lemmas \ref{lem:partial_derivs_|b|} and \ref{lem:.covariant_derivs_u} are given by tensorial expressions, which agree on chart overlaps, hence, they are globally well defined. Moreover, they give, up to the precise expression of $\tau$, the coordinate expressions of the differential constraints \eqref{eq:db} in point  3. of Theorem \ref{thm:metrizability_general_ab}.

\item Second, substituting the obtained relations $\mathcal{T} = \tau, \mathcal{S} = -\epsilon {\tau} $ into \eqref{def_S} and \eqref{def:T}, the metrizability conditions reduce to the PDE system
\end{enumerate}
    \begin{eqnarray}
- \epsilon \tau &&=-\dfrac{1}{2\Psi _{p}^{\prime }}\left[ \dfrac{\lambda }{p}
\Psi _{\left\vert b\right\vert }^{\prime }+\left\vert b\right\vert \left(
-c_{1}+\left\vert b\right\vert ^{2}c_{3}\left( p-\epsilon \right) \right)
\Psi _{p}^{\prime }-\left\vert b\right\vert \left( \left\vert b\right\vert
^{2}c_{3}+\dfrac{c_{1}}{p}\right) \right] \label{eq:first},   \\
\tau &&=-\dfrac{1}{2\Psi _{p}^{\prime }}\left\vert b\right\vert
\left( 2pc_{2}+\epsilon c_{1}-2\epsilon c_{2}\right) \label{eq:second} .
\end{eqnarray}
This is still an overdetermined system in the unknown $\Psi = \ln \Phi$, containing the free functions $\lambda=\lambda(|b|)$ and $\tau = \tau(|b|)$, as well as the constants $c_1,c_2,c_3$ as parameters. We distinguish two branches, depending on whether $\tau$ identically vanishes or not; its case by case integration will reveal the remaining conditions 1.-2. of Theorem \ref{thm:metrizability_general_ab}, as well as the precise expression of $\tau$.

\bigskip

\noindent \textbf{Branch 1: $\tau =0$}.  

In this case, recalling that $|b|$ cannot identically vanish, see the remark around \eqref{eq:phi_b}, we find ourselves in case (i) of Theorem \ref{thm:metrizability_general_ab}, as \eqref{eq:second} implies that:
\begin{equation}
c_1 = c_2 = 0,\,\,\text{hence,} \,\, c_3 \neq 0.
\end{equation}
We note that, since $c_1 =0$, the expression $\tau = 0 $ is consistent with relation \eqref{eq:tau_general}.
Moreover, the only nontrivial equation is the first one \eqref{eq:first}, which becomes
\begin{equation}\label{eq:PDE_case_c2_c_1_zero}
\lambda \Psi _{\left\vert b\right\vert }^{\prime }+\left\vert b\right\vert
^{3}pc_{3}\left( p-\epsilon \right) \Psi _{p}^{\prime }-\left\vert
b\right\vert ^{3}pc_{3}=0.
\end{equation}
The general solution of \eqref{eq:PDE_case_c2_c_1_zero} is
\begin{equation}
\Psi({|b|,p})=c_3 \epsilon \int \frac{|b|^3}{\lambda(|b|)} d|b| + \ln \left( \frac{p}{\epsilon} \right)+F\left( \frac{e^{-c_3 \epsilon \int \frac{|b|^3}{\lambda(|b|)}d|b|}  (\epsilon -p)}{p \epsilon} \right),
\end{equation}
where $F$ is a free function. Exponentiating leads to the desired solution
\begin{equation}\Phi(|b|,p)=  \frac{p}{\epsilon}  \exp\left( {c_3 \epsilon \int \frac{|b|^3}{\lambda(|b|)} d|b|} + F \left( \frac{e^{-c_3 \epsilon \int \frac{|b|^3}{\lambda(|b|)} d|b|}(\epsilon -p) }{p \epsilon}\right) \right) \; ,\end{equation}
which is precisely the claimed solution in case $(i)$ of Theorem \ref{thm:metrizability_general_ab}.

\bigskip

\noindent \textbf{Branch 2: $\tau \neq 0$}:

If $\tau \neq 0$, then the second equation \eqref{eq:second} can be directly integrated to give
\begin{equation}\label{eq:substitute_metrizabilitypsi}
    \Psi(|b|,p)=-\frac{1}{2} \frac{|b|}{\tau(|b|)} \left( p^2 c_2 + p \epsilon (c_1 - 2c_2) + k_1\right), \; \; 
\end{equation}
where $k_1 = k_1(|b|)$ is a free function. \\
This gives by $|b|$-differentiation:
\begin{equation}\label{eq:identifyp1}
    \Psi^{\prime}_{|b|}=\left(-\frac{1}{2 \tau^2} c_2 (\tau - |b| \tau ') \right)p^2+ \left(-\frac{\epsilon}{2 \tau^2} (\tau - |b| \tau^{\prime})(c_1-2c_2)\right)p+\left(-\frac{1}{2\tau^2} k_1 \left(\tau - |b| \tau^{\prime} \right) - \frac{|b|}{2\tau} k_1^{\prime} \right).
\end{equation}
On the other hand, $\Psi'_{|b|}$ can also be obtained from equation \eqref{eq:first}, by substituting $\Psi'_{p}$ from \eqref{eq:second}, as
\begin{equation}\label{eq:identifyp2}
\begin{aligned}
    \Psi^{\prime}_{|b|}=&p^3 \left(\frac{ |b|^4}{\lambda \tau} c_2 c_3 \right)+p^2 \left( -\frac{2 \epsilon |b| c_2}{\lambda}-\frac{|b|^2 c_1 c_2}{\tau \lambda} - \frac{2 \epsilon c_2 c_3 |b|^4}{\tau \lambda} + \frac{|b|^4}{2 \tau \lambda} c_1 c_3 \epsilon \right)\\
    +& p \left(- \frac{|b|}{\lambda} c_1 + 2 \frac{|b|}{\lambda} c_2 - \frac{|b|^2}{2 \lambda \tau} \epsilon  c_1^2 + \frac{|b|^2}{\tau \lambda}\epsilon c_1 c_2 - \frac{|b|^4}{2 \lambda \tau} c_3 c_1 + \frac{|b|^4}{\lambda \tau} c_3 c_2 + \frac{|b|^3}{\lambda} c_3\right) +\frac{|b|}{\lambda} c_1.
\end{aligned}
\end{equation}
Identifying the powers of $p$ in expressions \eqref{eq:identifyp1} and \eqref{eq:identifyp2} , taking into account that $\lambda \neq 0, \tau \neq 0, |b| \neq 0$, then gives
\begin{enumerate}
    \item[$\triangleright$] $p^3: c_2 c_3=0$.
    \item[$\triangleright$] $p^2: \frac{|b|}{2}\frac{-4 \tau \epsilon  c_2 - 2|b| c_1 c_2 -4 \epsilon c_2 c_3 |b|^3 + |b|^3 c_1 c_3 \epsilon}{\tau \lambda}+\frac{1}{2 \tau^2} c_2 (\tau - |b| \tau^{\prime})=0$.
    \item[$\triangleright$] $p: \frac{|b|}{2} \frac{-2c_1 \tau +4 c_2 \tau -|b| \epsilon c_1^2 + 2 |b| \epsilon c_1 c_2 -|b|^3 c_1 c_3 + 2 c_2 c_3 |b|^3 + 2 |b|^2 c_3 \tau}{\lambda \tau}+\frac{\epsilon}{2 \tau ^2} \left(\tau - |b| \tau^{\prime} \right) \left( c_1 -2 c_2 \right)=0.$
    \item[$\triangleright$] $p^0: \frac{|b|}{\lambda} c_1+ \left( \frac{k_1}{2 \tau^2} \left( \tau - |b| \tau^{\prime} \right) + \frac{|b|}{2 \tau} k_{1}^{\prime} \right)=0.$
\end{enumerate}
From the coefficient of $p^{3}$, we find that either $c_2=0$ or $c_3=0$.  This gives rise to two sub-branches:

\begin{enumerate}
    \item[\bf{2.a,}] {\textbf{$\mathbf{c_2 = 0}$:}}
    
    In this case, the restriction obtained from the coefficient of $p^2$ imposes $c_1 = 0$ or $c_3 = 0$.  It turns out, however, that the case $c_1=0$ leads, by the equation \eqref{eq:second},  to $\tau=0$, in contradiction with our hypothesis $\tau \neq 0$. Hence, if $c_2=0$, the only possibility is
    \begin{equation}
        c_3 = 0, \text{hence} \,\, c_1 \neq 0,
    \end{equation}
which implies the necessity of condition 2. in the statement of the theorem for this case.
    This corresponds to the case (ii). a) in the statement of the theorem.
    Using $c_2=c_3=0$, the (remaining) coefficients of $p$ and $p^0$ give the restrictions
    \begin{equation}
    \begin{cases}
        2 |b| \tau^2 - \epsilon \lambda \tau + |b|^2 \tau \epsilon c_1 
        + |b| \lambda \epsilon \tau' = 0, \\[6pt]
        k_{1}^{\prime} = -\dfrac{2}{|b|} 
        \left(
        \frac{k_1}{2\tau}\left(\tau - |b| \tau^{\prime} \right)
        + \frac{|b| \tau}{\lambda} c_1
        \right).
    \end{cases}
    \end{equation}
    The equation for $\tau$ is of Bernoulli-type and can be directly integrated to yield
    \begin{equation}\label{eq:tau_case_ii}
        \tau(|b|)=\frac{ c_1 |b|}{ c_1 C_1 e^{c_1 \rho(|b|)} - 2 \epsilon},
    \end{equation}
    which is precisely \eqref{eq:tau_general}. Substitution of  $\tau$ into the equation for $k_1$ plus a direct integration gives
    \begin{equation}\label{eq:k_1_case_ii}
        k_1(|b|)=\frac{C_2 - 2 c_1 \rho(|b|)}{C_1 e^{c_1 \rho(|b|)}-2\frac{\epsilon}{c_1}}.
    \end{equation}
    Substituting these values into equation \eqref{eq:substitute_metrizabilitypsi}  leads to the desired solution,
    \begin{equation}
        \Psi(|b|,p)=-\dfrac{1}{2}\dfrac{c_{1}C_{1}e^{c_{1}\rho(|b|) }-2\epsilon }{c_{1}}\left( p \epsilon c_1 + \frac{C_2- 2 c_1 \rho(|b|)}{C_1 e^{c_1 \rho(|b|)}-2 \frac{\epsilon}{c_1}} \right),
    \end{equation}
    which exponentiates to
    \begin{equation}
        \Phi(|b|,p)=\exp \left( \left(\frac{\epsilon}{c_1}- \frac{C_1}{2} e^{c_1 \rho(|b|)} \right) \left( p \epsilon c_1 + \frac{C_2  - 2 c_1 \rho(|b|)}{C_1 e^{c_1 \rho(|b|)}-2 \frac{\epsilon}{c_1}} \right) \right)
    \end{equation}
    as claimed for case $(ii)$ $(a)$.  
    
    \item[\bf{2.b,}] {\textbf{$\mathbf{c_2 \neq 0}$:}}
    
    In this case, the $p^3$-equation above implies 
    \begin{equation}
        c_3 =0.
    \end{equation}
    Further, the restriction obtained from the coefficient of $p^2$ gives an ordinary differential equation for $\tau$
    \begin{equation}\label{eq:conditionfromp^2_case_iii}
        \tau'=-\frac{4 \epsilon}{\lambda} \tau^2 + \left(\frac{1}{|b|} - \frac{2 |b| c_1}{\lambda} \right) \tau.
    \end{equation}
    Substituting this into the restriction obtained from $p^1$ gives the condition
    \begin{equation}
        (c_1-2c_2)\left(2 \tau \epsilon + |b| c_1 \right)=0.
    \end{equation}
    The above equation holds if either $2 \tau \epsilon + |b| c_1=0$, or $c_1=2c_2$.  This gives again, a splitting into two subcases. Again, we see that $c_1 =0$ is impossible, as it would imply either $\tau =0$ (in contradiction with our hypothesis), or $c_1 = c_2 =c_3 =0$, i.e., to trivial nonmetricity, which is excluded, This proves the necessity of Condition 2. in this case.
    
    \begin{enumerate}
        \item[\bf{2.b.I,}] {\textbf{$\mathbf{c_2 \neq 0,\ c_1\neq 2 c_2}$:}}
        
        In this case, we must necessarily have $2 \tau \epsilon + |b| c_1=0$, which gives the algebraic expression for $\tau$
        \begin{equation} \label{eq:tau_case_iii}
            \tau(|b|)=-\frac{|b| c_1}{2 \epsilon},
        \end{equation}
which is again, \eqref{eq:tau_general}  with $C_1 = 0$ and also solves the condition obtained from $p^2$ \eqref{eq:conditionfromp^2_case_iii}.  

Moreover, we note that equation  \eqref{eq:tau_case_iii} also implies that
        \begin{equation}
            \tau -|b| \tau'=0.
        \end{equation}Hence, in this case, we are left with the condition for $p^0$:
        \begin{equation}
            k_1^{\prime}=\frac{c_1^2 |b| }{\epsilon \lambda},
        \end{equation}
        which can be directly integrated to yield
        \begin{equation}\label{eq:k_1_case_iii}
            k_1(|b|)=\epsilon c_1^2 \rho(|b|) +C_3, \; \; \text{where} \;\; C_3 \; \; \text{is a real constant}.
        \end{equation}
        Finally, substituting the obtained value of $k_1$ into $k_1(|b|)$  into equation \eqref{eq:substitute_metrizabilitypsi} and exponentiating the result gives
        \begin{equation}
            \Phi(|b|,p)=\exp \left( \frac{\epsilon}{c_1} \left(p^2 c_2 +p \epsilon(c_1 - 2c_2) + \epsilon c_1^2 \rho(|b|)+C_3 \right) \right) ,
        \end{equation}
        which is precisely the claimed solution for case $(ii)$ $(b)$.
        
        \item[\bf{2.b.II,}] {\textbf{$\mathbf{c_2 \neq 0,\ c_1= 2 c_2}$:}}

        In this situation, recalling that $c_3=0$, the restriction obtained from the $p^1$ equation is identically satisfied. From the $p^2$ constraint we obtain
        \begin{equation}
            \tau '= -\frac{4 \epsilon}{\lambda} \tau^2 + \left(\frac{1}{|b|}-\frac{2 |b| c_1}{\lambda} \right) \tau,
        \end{equation}
        while the $p^{0}$ equation gives
        \begin{equation}
            k_1^{\prime}=\left(\frac{\tau^{\prime}}{\tau} - \frac{1}{|b|} \right) k_1 - \frac{2}{\lambda} \tau c_1.
        \end{equation}
        Directly integrating the above equations, we find
        \begin{equation}
            \tau(|b|)= \frac{|b| c_1}{-2 \epsilon+c_1 C_1 e^{2 c_1 \rho(|b|)}},  \; \; k_1(|b|)=\frac{c_1}{-2 \epsilon+c_1 C_1 e^{2 c_1 \rho(|b|)}} \left( - 2c_1 \rho(|b|)+C_4 \right), \; \; \text{with} \; \; C_1, C_4 \in \mathbb{R}. 
        \end{equation}
        Substituting these, together with $c_1-2c_2=0$ into equation \eqref{eq:substitute_metrizabilitypsi} and exponentiating leads to the claimed solution for case $(ii)$ $(c)$:
        \begin{equation}
            \Phi(|b|,p)=\exp \left( \left(\frac{\epsilon}{2} - \frac{C_1}{2} c_2 e^{4 c_2 \rho(|b|)} \right)p^2 + \left(2 c_2 \rho(|b|)-\frac{C_4}{2} \right) \right) .
        \end{equation}
    \end{enumerate}
\end{enumerate}

The above case-by-case analysis also showed that we must necessarily have $c_3 = 0$ or $c_1 = c_2 =0$, as announced, as well as the expression \eqref{eq:tau_general} for $\tau$, where $C_1 = 0$ in the case where $c_1 - 2c_2 \neq 0$, which proves the last part of the necessity. \\
Noting that, for generalized $(\alpha,\beta)$ metrics, the determinant is still given by a similar formula to \eqref{eq:detg}, we determine by direct computation that all the solutions found in case $(ii)$ are nondegenerate, while non-degeneracy imposes a constraint (expressed as an ODE) on the free function $F$ in case $(i)$: yet, apart from the solutions of the respective ODE, any choice of $F$ is valid. Moreover, in each of the analyzed cases, we have explicitly shown that a solution to the metrizability conditions \eqref{def:Berwald_conditions} exists, meaning that the given conditions are also sufficient for the Finsler metrizability of the respective connections.

These arguments complete the proof of our main Theorem \ref{thm:metrizability_general_ab}.

\section{Discussion}\label{sec:Discussion}
In this paper we answered the question  under which conditions the autoparallels of affine connections with nonmetricity are extremals of an action functional, for
torsion-free affine connections with vectorial nonmetricity (meaning that the distortion tensor is expressed algebraically in terms of the metric and of a one-form~$b$). 

Specifically, we asked the question of whether their autoparallels are actually, (pseudo-)Finslerian geodesics 'in disguise', a feature known as \textit{Finsler metrizability}; the relevance of such a question is, on the one hand from a mathematical perspective, given by the equivalence between Finsler metrizability and the existence of a parametrization-invariant variational principle for autoparallels, and, on the other hand from a physical perspective, given by the question on the interpretation of autoparallels in metric-affine gravity. 

\bigskip

Our findings are summarized in Theorems \ref{thm:metrizability_classification_alpha_beta} and \ref{thm:metrizability_general_ab}. We find the necessary and sufficient conditions for a connection with vectorial nonmetricity to be metrizable by a Finsler function depending algebraically on its constituents (i.e., on the metric and on the defining one-form $b$), together with the most general form of the respective Finsler functions. In other words, we find the action principle that is extremised by the autoparallels of affine connections with vectorial nonmetricity. The Finsler functions that define the action belong to the so-called class of \textit{generalized $(\alpha,\beta)$-metrics}, admitting as a particular case, the most commonly used in applications class of Finsler functions, namely,  $(\alpha,\beta)$-metrics. 
It turned out that, provided that the one-form $b$ satisfies some specific differential constraints (in particular, it is torse-forming \cite{gherici2026torseformingvectorfieldcertain} and closed), then quite a large class of connections built by means of it are actually, Finsler metrizable. Below is a table indicating the situation for the main examples of connections with vectorial nonmetricity - the famous classes of Weyl, Schrödinger and completely symmetric connections. 

\begin{table}[h!]
\centering
\begin{tabular}{|c|c|c|}
\hline
Connection & ($\alpha,\beta)$-metrizable & Generalized $(\alpha,\beta)$-metrizable\\ \hline
Weyl $(c_1 \neq 0, c_2=0, c_3=0)$ & {\color{green}\ding{51}} & {\color{green}\ding{51}}\\ \hline
Schrödinger $(c_3=0,c_1+2c_2=0)$ & {\color{red}\ding{55}} & {\color{green}\ding{51}} \\ \hline
Completely symmetric $(c_1=c_2)$ & {\color{green}\ding{51}} & {\color{green}\ding{51}} \\ \hline
\end{tabular}
\end{table}

Since we found an explicit relation between Berwald-Finsler structures and metric-affine geometry, a very interesting future direction of research is to use the above found Finsler functions as ansatzes in Finsler gravity and to compare the resulting field equations with the ones usually employed in metric-affine gravity. Are they related, or do the Finsler gravity equations lead to new, not yet studied dynamics? In particular, the Finsler gravity equation naturally couples the 1-particle distribution function of kinetic gases to spacetime geometry, this line of research will lead to a new perspective on gravity theories with nonmetricity and the matter model that sources these geometries. Concrete interesting physical systems we plan to study are compact objects and the evolution of the universe, to investigate if these geometries can give further insights on an explanation of dark matter and dark energy.

\acknowledgments{L.Cs. is grateful for the support of Collegium Talentum, Hungary. C.P. acknowledges support by the excellence cluster QuantumFrontiers of the German Research Foundation (Deutsche Forschungsgemeinschaft, DFG) under Germany's Excellence Strategy -- EXC-2123 QuantumFrontiers -- 390837967 and was funded by the Deutsche Forschungsgemeinschaft (DFG, German Research Foundation) - Project Number 420243324. L.Cs., C.P. and N.V. would like to acknowledge networking support by the COST Actions CA23130
“Bridging high and low energies in search of quantum gravity (BridgeQG)” and CA21136
“Addressing observational tensions in cosmology with systematics and fundamental physics
(CosmoVerse)”. }
\bibliographystyle{unsrt}
\bibliography{metrize}

\end{document}